\documentclass[11pt]{article}
\usepackage[latin1]{inputenc}
\usepackage{fullpage}   %
\usepackage{graphicx}  %
\usepackage{enumitem}
\usepackage{mleftright}
\usepackage{comment}
\usepackage{xspace}
\usepackage{nameref}
\usepackage{hyperref}
\usepackage{cleveref}
\usepackage{color}
\newcommand{\alert}[1]{\textcolor{black}{#1}}
\usepackage{ifthen}
\usepackage{algorithm}
\usepackage[noend]{algpseudocode}
\usepackage{wrapfig}

\newboolean{short}
\setboolean{short}{false}

\newcommand{\shortOnly}[1]{\ifthenelse{\boolean{short}}{#1}{}}
\newcommand{\onlyShort}[1]{\ifthenelse{\boolean{short}}{#1}{}}
\newcommand{\longOnly}[1]{\ifthenelse{\boolean{short}}{}{#1}}
\newcommand{\onlyLong}[1]{\ifthenelse{\boolean{short}}{}{#1}}
\newcommand{\shortLong}[2]{\ifthenelse{\boolean{short}}{#2}{#1}}
\newcommand{\longShort}[2]{\ifthenelse{\boolean{short}}{#2}{#1}} 

\usepackage{amsmath}   %
\usepackage{amsfonts}   %
\usepackage{amssymb}   %
\usepackage{amsthm}     %
\usepackage{amsbsy}     
\usepackage{mathrsfs}

\DeclareMathOperator*{\argmax}{arg\,max}

\theoremstyle{plain}
\newtheorem{theorem}{Theorem}
\newtheorem{proposition}{Proposition}
\newtheorem{corollary}{Corollary}
\newtheorem{lemma}{Lemma}
\newtheorem{claim}{Claim}

\newtheorem{observation}{Observation}

\makeatletter
\DeclareRobustCommand*\cal{\@fontswitch\relax\mathcal}
\makeatother
\def\polylog{\operatorname{polylog}}
\def\cA{\mathcal{A}}

\def\cP{\mathcal{P}}

\def\cG{\mathcal{G}}

\def\bp{\mathbf{p}}

\let\originalleft\left
\let\originalright\right
\renewcommand{\left}{\mathopen{}\mathclose\bgroup\originalleft}
\renewcommand{\right}{\aftergroup\egroup\originalright}

\newcommand{\Prob}[1]{\text{Pr}\left[#1\right]\xspace}

\newcommand{\expect}[1]{\mathbb{E}\left[#1\right]\xspace}

\newcommand{\entropy}[1]{\mathbb{H}\left[#1\right]\xspace}

\newcommand{\information}[1]{\mathbb{I}\left[#1\right]\xspace}

\newcommand{\tokens}{\textsf{tokens}}

\newcommand{\infcost}{\mathrm{{IC}}\xspace}
\newcommand{\prt}{\mathcal{GP}\xspace}
\newcommand{\rand}{\mathcal{RS}\xspace}
\newcommand{\good}{\textsl{Good}\xspace}
\newcommand{\bad}{\textsl{Bad}\xspace}
\newcommand{\ra}{\rightarrow\xspace}
\newcommand{\bal}{\textsl{Bal}\xspace}
\newcommand{\out}{\textsl{Out}\xspace}
\newcommand{\poly}{\operatorname{poly}}

\newcommand{\pagerank}{{\sf PageRank}\xspace}

\newcommand{\congest}{\textsc{CONGEST}\xspace}

\renewcommand{\paragraph}[1]{\medskip\noindent{\bf #1.}\xspace}

\renewcommand{\le}{\leqslant}
\renewcommand{\ge}{\geqslant}
\renewcommand{\geq}{\geqslant}
\renewcommand{\leq}{\leqslant}

\title{On the Distributed Complexity of Large-Scale Graph Computations}

\author{
Gopal Pandurangan\thanks{Department of Computer Science, University of Houston, Houston, TX 77204, USA.
\hbox{E-mail}:~{\tt gopalpandurangan@gmail.com}.
Supported, in part, by NSF grants CCF-1527867, CCF-1540512, IIS-1633720, CCF-1717075, and by  BSF grants 2008348 and 2016419.}
\and
Peter Robinson\thanks{Department of Computing and Software, McMaster University, Hamilton, Canada L8S 4L7.
\hbox{E-mail}:~{\tt peter.robinson@mcmaster.ca}.}
\and Michele Scquizzato\thanks{ School of Electrical Engineering and Computer Science, KTH Royal Institute of Technology, Sweden.
\hbox{E-mail}:~{\tt mscq@kth.se}. Supported, in part, by the European Research Council (ERC) under the European Union's
Horizon 2020 research and innovation programme under grant agreement No 715672.}}
\date{}

\begin{document}

\maketitle
\thispagestyle{empty}

\begin{abstract}
Motivated by the increasing need to understand the distributed algorithmic foundations of large-scale graph
computations, we study some fundamental graph problems in a message-passing model for distributed
computing where $k \geq 2$ machines jointly perform computations on graphs with $n$ nodes (typically, $n \gg k$).
The input graph is assumed to be initially randomly partitioned among the $k$ machines, a common
implementation in many real-world systems. Communication is point-to-point, and the goal is to minimize
the number of communication {\em rounds} of the computation.

Our main contribution is the {\em General Lower Bound Theorem}, a  theorem that can be used to show non-trivial lower bounds on the round complexity of distributed large-scale data computations. The General Lower Bound Theorem is established via an information-theoretic approach that relates the round complexity to the minimal amount of information required by machines to solve the problem. Our approach is generic and  this theorem can be used in a ``cookbook" fashion to show distributed lower bounds in the context of several problems, including non-graph problems. We present two applications by showing
(almost) tight lower bounds for the round complexity of two fundamental graph problems, namely {\em PageRank computation} and {\em triangle enumeration}. Our approach, as demonstrated in the case of PageRank, can yield tight lower bounds for problems (including, and especially, under a stochastic partition of the input) where communication complexity techniques are not obvious.
 Our approach, as demonstrated in the case of triangle enumeration, can yield stronger round lower bounds as well as message-round tradeoffs compared to approaches that use communication complexity techniques.
 
We then present distributed algorithms for PageRank and triangle enumeration with a round complexity that (almost) matches the respective lower bounds; these algorithms exhibit a round complexity which scales superlinearly in $k$, improving significantly over previous results for these problems [Klauck et al., SODA 2015].
Specifically, we show the following results:
\begin{itemize}
\item {\em PageRank:} We show a lower bound of $\tilde{\Omega}(n/k^2)$ rounds, and present a distributed
algorithm that computes an approximation of the PageRank of all the nodes of a graph in $\tilde{O}(n/k^2)$ rounds.
\item {\em Triangle enumeration:} We show that there exist graphs with $m$ edges where any distributed algorithm
requires  $\tilde{\Omega}(m/k^{5/3})$ rounds.
This result also implies the first non-trivial lower bound of $\tilde\Omega(n^{1/3})$ rounds for the
{\em congested clique} model, which is tight up to logarithmic factors.
We then present a distributed algorithm that enumerates all the triangles of a graph in $\tilde{O}(m/k^{5/3} + n/k^{4/3})$ rounds. 
\end{itemize}
\end{abstract}
\maketitle

\onlyShort{\vspace{-0.1in}}
\section{Introduction}
\onlyShort{\vspace{-0.05in}}
The focus of this paper is on the distributed processing of large-scale data, in particular,  \emph{graph data},
which is becoming increasingly important with the rise of massive graphs such as the Web graph, social networks, biological networks,
and other graph-structured data and the consequent need for fast distributed algorithms to 
process such graphs. Several large-scale graph processing systems such as Pregel~\cite{pregel} and Giraph~\cite{giraph}  have been recently
designed based on the {\em message-passing} distributed computing model~\cite{Lynch96,Peleg00}. In these systems, the input graph, which is simply too large to fit into a single machine, is distributed across a group of machines
that are connected  via a communication network and the machines jointly perform  computation in a distributed fashion by sending/receiving messages. A key goal in distributed large-scale computation is to minimize the amount of communication across machines, as this typically dominates the overall cost of the computation. \onlyLong{Indeed, Reed and Dongarra in a recent CACM article \cite{cacm} on distributed Big Data computing emphasize:
``It is important for all of computer science to design algorithms that communicate as little as possible, ideally attaining lower bounds on the amount of communication required."}

We study fundamental graph problems in a message-passing distributed computing model 
and present almost tight bounds on the number of communication rounds needed to solve these problems.
In the model, called the $k$-machine model~\cite{KlauckNPR15} (explained in detail in Section~\ref{sec:model}), the input graph (or more generally, any other type of data)  is distributed across a group of
$k$ machines that are pairwise interconnected via a communication network. 
The $k$ machines jointly perform computations on an arbitrary $n$-vertex input graph (where typically $n \gg k$) distributed
among the machines.
 The communication
is point-to-point via message passing. 
The goal is to minimize the {\em round  complexity},
i.e., the number of \emph{communication rounds}, given some (bandwidth) constraint on the amount of data that each
link of the network can deliver in one round.
 We address a fundamental issue in distributed computing of large-scale data:  What is the distributed (round) complexity of solving problems when each machine can see only {\em a portion of the input} and there is a {\em limited bandwidth} for communication? 
We would like to quantify the round complexity of solving problems as a function
of the {\em size of the input} and the {\em number of machines} used in the computation. 
\onlyLong{In particular, we would like to quantify how the  round complexity 
scales with the number of machines used:
more precisely, does the number of rounds scale
linearly (or even super-linearly) in $k$? 
And what is the best possible round complexity  for various  problems?}

A main contribution of this paper is a technique that can be used to show non-trivial lower bounds on the distributed complexity
(number of communication rounds) of large-scale data computations, and its application to graph problems.

\onlyShort{\vspace{-0.1in}}
\subsection{The Model}
\onlyShort{\vspace{-0.05in}}
\label{sec:model}

We now describe the adopted model of distributed computation, the \emph{$k$-machine model}
(a.k.a.\ the \emph{Big Data model}), introduced in~\cite{KlauckNPR15} and further investigated
in~\cite{spidal,fanchung,PanduranganRS16,BandyapadhyayIPP18,PanduranganRS18}.
The model consists of a set of $k \geq 2$ machines $\{M_1,M_2,\dots,M_k\}$ 
that are pairwise interconnected by bidirectional point-to-point communication links.
Each machine executes an instance of a distributed algorithm. The computation advances
in synchronous rounds where, in each round, machines can exchange messages over their
communication links and perform some local computation. Each link is assumed to have
a bandwidth of $B$ bits per round, i.e., $B$ bits can be
transmitted over each link in each round;  unless otherwise stated, we assume $B = \Theta(\polylog n)$.\footnote{There is an alternative (but equivalent) way to view this communication
restriction: instead of putting a bandwidth restriction on the links, we can put a restriction on
the amount of information that each \emph{machine} can communicate (i.e., send/receive)
in each round.  The results that we obtain in the bandwidth-restricted model will also apply
to the latter model~\cite{KlauckNPR15}. Also, our bounds can be easily rewritten in terms of the $B$ parameter.}  Machines do not share any memory and have no other
means of communication. We assume that each machine has access to a private source of true random bits.
We say that algorithm $\cA$ has \emph{$\epsilon$-error} if, in any run of $\cA$, the output of the machines corresponds to a correct solution with probability at least $1 - \epsilon$.
To quantify the performance of a randomized (Monte Carlo) algorithm $\cA$, we define the \emph{round complexity of $\cA$}  to be
the worst-case number of rounds required by any machine when executing $\cA$. 

Local computation within a machine is considered to happen instantaneously at zero cost, while
the exchange of messages between machines is the costly operation.\onlyLong{\footnote{This assumption
is standard in the context of large-scale data processing.
Indeed, even assuming communication links with a bandwidth of order of gigabytes
per second, the amount of data that typically has to be communicated can be in the order of tera- or
peta-bytes, which generally dominates the overall computation cost~\cite{LeskovecRU14}.}}
However, we note that in
all the algorithms of this paper, every machine in every round performs lightweight computations; in particular, these computations are bounded by a polynomial (typically, even linear) in the size of the input assigned to that machine. 

Although the $k$-machine model is a general model of distributed computation that can be applied to study any (large-scale data)
problem, in this paper we focus on investigating
graph problems in it. 
Specifically, we are given an input graph $G$ with $n$ vertices, each associated with a unique
\alert{integer ID from $[n]$}, and $m$ edges. To avoid trivialities, we will assume that $n \geq k$ (typically, $n \gg k$).
Initially, the entire graph $G$  is not known by \alert{any} single machine, but rather partitioned  among  the  $k$ machines in a
``balanced'' fashion, i.e., the nodes and/or edges of $G$ must be partitioned approximately evenly among the machines.
We assume a {\em  vertex-partition} model, whereby vertices (and their incident edges) are partitioned across machines.
Specifically, the type of partition that we will assume throughout is the {\em random vertex partition (RVP)},
i.e., vertices (and their incident edges) of the input graph are assigned randomly to machines. This is
the typical way used by many real graph processing systems, such as Pregel~\cite{pregel} and Giraph~\cite{giraph,facebook}, to partition the input graph
among the machines; it is easy to accomplish, e.g., via hashing.

More formally, in the \emph{random vertex partition} model  each vertex of $G$  is assigned independently and uniformly at
random to one of the $k$ machines.\onlyLong{\footnote{An
alternate partitioning model,  called  the \emph{random edge partition (REP)} model has also been studied~\cite{woodruff,PanduranganRS16}: here, each edge of $G$ is
assigned independently and randomly to one of the $k$ machines.  One can extend our results to get bounds for the REP model since it is
easy to show that one can transform the input partition from one model to the other in $\tilde{O}(m/k^2 + n/k)$ rounds.}} If a vertex $v$ is assigned to machine $M_i$ we say that $M_i$ is the {\em home machine}
of $v$ and, with a slight abuse of notation, write $v \in M_i$. When a vertex is assigned to a machine, all its
incident edges
 are known to that machine as well, i.e., the home machine initially knows the IDs
of the neighbors of that vertex as well as the identities of their home machines (and
the weights of the corresponding edges in case $G$ is weighted). For directed graphs, we assume that out-edges of
vertices are known to the assigned machine. (However, we note that our lower bounds hold even if both in- and out-edges are known to the home machine.)
An immediate property of the RVP model is that the number of vertices at each machine is \emph{balanced},
i.e., each machine is the home machine of $\tilde \Theta(n/k)$ vertices with high probability (see~\cite{KlauckNPR15}); we shall assume this
throughout the paper. A convenient way to implement  the RVP model is through hashing: each vertex (ID)
is hashed to one of the $k$ machines. Hence, if a machine knows a vertex ID, it also knows where it is hashed to. 

Eventually,  in a computation each machine $M_i$, for each $1 \leq i \leq k$, must set a designated local output variable $o_i$ (which need
not depend on the set of vertices assigned to machine $M_i$), and the \emph{output configuration} $o=\langle o_1,\dots,o_k\rangle$
must satisfy certain feasibility conditions w.r.t.\ problem $\cP$. 
For example, when considering the PageRank problem,
each $o_i$ corresponds to  PageRank values of (one or more nodes), such that the PageRank value of each node of the graph
should be output by {\em at least one} (arbitrary) machine.

\onlyShort{
\vspace{-0.1in}
\subsection{Our Results}
\vspace{-0.05in}
}
\onlyLong{
\subsection{Our Results}
}
\label{sec:results}
We present a general information-theoretic approach for showing non-trivial round
lower bounds for certain graph problems in the $k$-machine model. This approach can be useful in the context of showing
round lower bounds for many other (including non-graph)  problems in a distributed setting where the input is
partitioned across several machines and the output size is large.
Using our approach we show almost tight (up to logarithmic factors) lower bounds for two fundamental, seemingly unrelated,
graph problems, namely PageRank computation and triangle enumeration.\onlyShort{\footnote{For a brief description, background, and significance of these problems, we refer to the full paper (in Appendix).}} These lower bounds apply to distributed computations in essentially all point-to-point communication models, since they apply even to a synchronous complete network model (where $k = n$), and {\em even} when the input is partitioned {\em randomly}, and thus they apply to worst-case balanced partitions as well (unlike some previous lower bounds, e.g., \cite{woodruff}, which apply only under some worst-case  partition). 
 
To demonstrate the near-tightness of our lower bounds we present optimal (up to a $\polylog(n)$ factor) distributed algorithms for these problems.
All these algorithms exhibit a round complexity that scales  {\em superlinearly} in $k$, improving significantly over previous results. 

\noindent {\bf 1. PageRank Computation.}  In Section~\ref{sec:prlower} we show an almost tight  lower bound of
$\tilde{\Omega}(n/k^2)$ rounds.\footnote{Notation $\tilde \Omega$ hides a $1/\text{polylog}(n)$ factor,
and $\tilde O$ hides a $\text{polylog}(n)$ factor and an additive $\text{polylog}(n)$ term.}
In Section~\ref{sec:pralgo} we present an algorithm that computes the
PageRank of all nodes of a graph in $\tilde{O}(n/k^2)$ rounds, thus improving over the previously known bound of $\tilde{O}(n/k)$ rounds~\cite{KlauckNPR15}.

\noindent {\bf 2. Triangle Enumeration.} 
In Section~\ref{sec:trianglelb} we show that there exist graphs with $m$ edges where any
distributed algorithm requires  $\tilde{\Omega}(m/k^{5/3})$ rounds. %
In Section~\ref{sec:trialgo} we present an algorithm that enumerates all the triangles of a graph in $\tilde{O}(m/k^{5/3} + n/k^{4/3})$ rounds.
This improves over the previously known bound of $\tilde{O}(n^{7/3}/k^2)$ rounds~\cite{KlauckNPR15}.

%
%

Our technique can be used to derive lower bounds in other models of distributed computing as well.
Specifically, the approach used to show the lower bound for triangle enumeration can be adapted
for the popular congested clique model (discussed in Section~\ref{sec:related}), yielding an
$\Omega(n^{1/3}/\log n)$ lower bound for the same problem.\footnote{A preliminary version
of this paper, appeared on arXiv~\cite{tightbounds}, contained a slightly worse lower bound of the form
$\Omega(n^{1/3}/\log^3 n)$; later, a subsequent work by Izumi and Le Gall~\cite{IzumiG17} showed
a lower bound of the form $\Omega(n^{1/3}/\log n)$ using our information-theoretic approach.}
(Notice that this does not contradict the impossibility result of~\cite{DruckerKO14}, which states
that any super-constant lower bound for the congested clique would give new lower bounds in circuit complexity:
because of the size required by any solution for triangle enumeration, Remark~3 in~\cite{DruckerKO14} does not apply.)
To the best of our knowledge, this is the first {\em super-constant} lower bound known for the congested clique model.
(Previous bounds were known for weaker versions of the model, e.g., which allowed only broadcast communication, or which
applied only to deterministic algorithms~\cite{DruckerKO14}, or for implementations of specific algorithms~\cite{Censor-HillelKKLPS15}.)

Our bounds for triangle enumeration also apply to the problem of enumerating all the \emph{open triads}, that is,
all the sets of three vertices with exactly two edges. Our techniques and results can be generalized to the enumeration of
other small subgraphs such as cycles and cliques.

%
%

%
%
%
%

%
%
%
%
%
%
%
%
%
%
\onlyShort{Due to lack of space, full proofs and additional details are deferred to the full paper (in Appendix). 
}

\onlyShort{
\vspace{-0.1in}
\subsection{Overview of Techniques}
\vspace{-0.05in}
}
\onlyLong{
\subsection{Overview of Techniques}
}
\label{sec:technical}

\noindent {\bf Lower Bounds.} In \Cref{thm:expectedruntime} we prove a  general result, the {\em General Lower Bound Theorem}, which relates the round complexity in the $k$-machine model to the minimal amount of information required by machines for correctly solving a problem.
While $\pagerank$ and triangle enumeration are fundamentally different problems, we derive lower bounds for both problems via the ``information to running time'' relationship of \Cref{thm:expectedruntime}. 
The  General Lower Bound Theorem   gives  two probabilistic bounds
that must be  satisfied in order to obtain a lower bound on the round complexity of any problem. The two bounds together capture
the decrease in uncertainty (called {\em surprisal}, see Section~\ref{sec:lb}) that happens to some machine as a result
of outputting the solution. 
We can show that this ``surprisal change''  represents the maximum expected ``Information Cost"  over all machines which 
can be used to lower bound the run time. The proof of the General Lower Bound Theorem makes use of information-theoretic machinery,
yet its application requires no use of information theory.

We conjecture that \Cref{thm:expectedruntime} can be used to obtain lower bounds for various problems (including non-graph problems) that
have a relatively {\em large output size} (e.g.,  shortest paths, sorting, matrix multiplication, etc.) thus complementing the approach based on communication complexity (see, e.g., \cite{peleg-bound,sicomp12,podc11,podc14,oshman-survey,DruckerKO14,KlauckNPR15,PanduranganRS16,PanduranganPS16}
and references therein). In fact, our approach, as demonstrated in the case of triangle enumeration, can yield stronger round lower bounds as well as message-round tradeoffs compared to approaches that use communication complexity techniques (more on this in the next paragraph). Our approach, as demonstrated in the case of PageRank, can yield tight lower bounds for problems (including, and especially, under a stochastic/random partition of the input) where communication complexity techniques are not obvious. In fact, for many problems, applying the General Lower Bound Theorem gives non-trivial lower bounds in a fairly straightforward way  that are  not   (at least easily) obtainable by communication complexity techniques. To give an example, the work of Klauck et al. \cite{KlauckNPR15} showed a  lower bound of $\tilde{\Omega}(n/k^2)$ for connectivity by appealing to random partition communication complexity---this involved proving the classical set disjointness lower bound {\em under random input partition}, which involved non-trivial work. On the other hand, the same lower bound of  $\tilde{\Omega}(n/k^2)$ for MST can be shown directly\footnote{The lower bound graph can be a complete graph with random edge weights.} via the General Lower Bound Theorem (this bound is tight due to the algorithm of \cite{PanduranganRS16}).
To give another example, consider the problem of distributed  sorting (see, e.g., \cite{PanduranganPS16}), whereby $n$ elements are randomly distributed across  the $k$ machines and the requirement is that, at the end, the $i$-th machine must hold the $(i-1)k +1, (i-1)k +2,\dots, i \cdot k$-th order statistics. One can use the General Lower Bound Theorem to show a $\tilde{\Omega}(n/k^2)$ lower bound for this problem (and this is tight, as there exists an $\tilde{O}(n/k^2)$-round sorting algorithm). Note that the same lower  bound (under a random partition) is harder to show using communication 
complexity techniques.\footnote{Assuming an adversarial (worst-case) balanced partition (i.e., each machine gets $n/k$ elements), using multi-party communication complexity techniques one can show
the same lower bound \cite{PanduranganPS16}; but this is harder to show under random partition.}

We also note that tight {\em round} complexity lower bounds do not always directly follow from exploiting {\em message (bit)} complexity lower bounds obtained by leveraging communication complexity results. For example, for the problem of triangle enumeration, even assuming the highest possible message lower bound of $\Omega(m)$, this would directly imply a {\em round} lower bound of $\tilde{\Omega}(m/k^2)$ (since $\Theta(k^2)$ messages can be exchanged in one round) and not the tight $\tilde{\Omega}(m/k^{5/3})$ shown in this paper. Furthermore, our approach can show round-message {\em tradeoffs} giving stronger message lower bounds for algorithms
constrained to run in a prescribed round bound compared
to what one can obtain using  communication complexity approaches. In particular, for triangle enumeration, we show that any round-optimal  algorithm  that enumerates all triangles with high probability  in the $k$-machine model needs to exchange  a total of  $\tilde\Omega(m k^{1/3})$ messages in the worst case. 

We emphasize that our General Lower Bound theorem gives non-trivial lower bounds only when the output size is large enough, but it still
works seamlessly across all output sizes. To illustrate this, we note that the triangle enumeration lower bound
of $\tilde{\Omega}(m/k^{5/3})$ is true only for dense graphs, i.e., $m = \Theta(n^2)$.  In fact, the real lower bound derived through
our theorem is $\tilde{\Omega}((t/k)^{2/3}/k)$, where $t$ is the number of triangles in the input graph; this bound can be shown to apply even for sparse (random) graphs by extending our analysis.

Entropy-based information-theoretic arguments have been used in prior work~\cite{KlauckNPR15}.
However, there is a crucial difference, as explained next.
In~\cite{KlauckNPR15}, it was shown that $\tilde{\Omega}(n/k)$ is a lower bound for computing a  spanning tree (ST) of
a graph. However, this lower bound holds
under  the criterion that the machine which hosts the vertex (i.e., its home machine) must know
at the end of the computation the status of all of its incident edges (whether they belong to a ST or not) and output
their respective status. 
The lower bound proof exploits this criterion to show that any algorithm will require some machine receiving
$\Omega(n)$ bits of information, and since any machine has $k-1$ links, this gives a $\tilde \Omega(n/k)$ lower bound.
This argument fails if we require the final status of each edge to be known by \emph{some}
machine (different machines might know the status of different edges); indeed  under this output criterion,  it can be shown that
MST can be solved in $\tilde{O}(n/k^2)$ rounds~\cite{PanduranganRS16}. On the other hand, the lower bound proof technique of this paper applies to the less restrictive (and more natural)  criterion that any machine can output  any part of the solution.
In \cite{DBLP:journals/jcss/Bar-YossefJKS04}, a direct sum theorem is shown that yields a communication complexity lower bound for set disjointness. The method of \cite{DBLP:journals/jcss/Bar-YossefJKS04} can be applied to obtain lower bounds for functions $F$ that can be ``decomposed'' as $F(\mathbf{x},\mathbf{y}) = f(g(x_1,y_1),\dots,g(x_n,y_n))$, by
reduction from the information complexity of the function $g$. These methods do not seem applicable to our setting as we are considering problems where the output size is large.

\noindent {\bf Upper Bounds.} 
The Conversion Theorem of~\cite{KlauckNPR15} directly translates algorithms designed for a message passing model for network algorithms
to the $k$-machine model, and almost all the previous algorithms~\cite{KlauckNPR15,fanchung,spidal} were derived
using this result. In contrast, the present paper does not use the Conversion Theorem; instead, it gives direct solutions
for the problems at hand in the $k$-machine model, leading to improved algorithms with significantly better round complexity.

While our algorithms use techniques specific to each problem, we point out a simple, but key, unifying technique that proves very useful in designing fast algorithms,  called  {\em randomized proxy computation}.\footnote{Similar ideas have been used in parallel and distributed computation in different contexts, see, e.g.,~\cite{Valiant82,S98}.}   Randomized proxy computation is crucially used to distribute communication {\em and} computation across machines
to avoid congestion at any particular machine, which instead is redistributed evenly across all the machines. This is achieved, roughly speaking, by re-assigning the {\em executions} of individual nodes uniformly at random among the machines.  Proxy computation allows one to move away from the
communication pattern imposed by the topology of the input graph, which can cause congestion at a particular machine,
to a more balanced communication overall.
\onlyLong{For example, a simple use of this strategy in the triangle enumeration algorithm (see Section~\ref{sec:trialgo}) is as follows: each edge in the graph is assigned a random machine as its proxy; the proxy does computation ``associated" with the edge. This alleviates the congestion associated with machines having high-degree nodes. A slightly more sophisticated use of randomized proxy computation is made
in our PageRank algorithm (see Section \ref{sec:pralgo}).}

\onlyShort{
\vspace{-0.1in}
\subsection{Related Work and Comparison}
\vspace{-0.05in}
}
\onlyLong{
\subsection{Related Work}
}
\label{sec:related}
\onlyLong{
The theoretical study of large-scale graph computation in distributed systems is relatively new.
Several works have been devoted to developing MapReduce graph algorithms
(e.g., see~\cite{lin-book,filtering-spaa,LeskovecRU14,soda-mapreduce,andoni} and references therein).  
We note that  the flavor of theory developed for MapReduce is quite different compared to this paper.
Minimizing communication  is also a key motivation in MapReduce algorithms (e.g., see \cite{LeskovecRU14}); however this is
generally achieved by making sure that the data is made small enough  quickly
to fit into the memory of a single machine, such as in the MapReduce algorithm of~\cite{filtering-spaa} for MST.\footnote{We note that in the $k$-machine model the memory usage is also implicitly captured. For example, consider
the $\pagerank$ algorithm of this paper. Each machine starts with a $1/k$ fraction of the input size (i.e., $\tilde{O}((m+n)/k + \Delta)$), and since the algorithm takes $\tilde{O}(n/k^2)$ rounds, the total number of messages received by a machine during the entire execution of the algorithm is $\tilde{O}(n/k)$. Furthermore, since the local computation uses only $\tilde{O}(n/k)$ space (i.e., essentially linear in the size of the input restricted to that machine), the overall memory used remains the same as the initial input to the machine.}
}

For a  comparison of the $k$-machine model  with other parallel and distributed models proposed for large-scale data processing, including Bulk Synchronous Parallel (BSP) model~\cite{bsp},  MapReduce~\cite{soda-mapreduce}, and the congested clique, we refer to \cite{grigory-blog}. 
In particular, according to~\cite{grigory-blog},
``Among all models with restricted communication the ``big data'' [$k$-machine] model is the one most similar to the MapReduce model".  \onlyShort{More details on related work (in particular, comparison of the $k$-machine model with other models), can be found in the full paper.} 

Klauck et al.~\cite{KlauckNPR15} present lower and upper bounds for several fundamental graph
problems in the $k$-machine model. In particular, they presented weaker upper bounds
for PageRank and triangle verification (which also works for triangle enumeration), which are substantially improved in this paper.
They do not present any non-trivial lower bound for any of these problems. Also, as pointed out earlier, some lower bounds shown in~\cite{KlauckNPR15}, most notably the $\Omega(n/k^2)$ lower bound of MST (under random input partition and under the requirement that each MST edge has to be output by {\em some} machine), can be shown in a simpler way using the General Lower Bound Theorem of this paper.
Pandurangan et al.~\cite{PanduranganRS16} showed  $\tilde{O}(n/k^2)$-round algorithms in the $k$-machine model for connectivity, MST, approximate min-cut,  and other graph verification problems. The algorithmic techniques used in that paper  (except for the randomized proxy computation) cannot be applied for PageRank and triangle enumeration.

\onlyLong{
The $k$-machine model is closely related to the BSP model \cite{bsp}; it can be considered to be a simplified version of BSP, where local computation is ignored and synchronization happens
at the end of every round (the synchronization cost is  ignored). Unlike BSP which has a lot of different parameters (which typically makes it harder to prove rigorous theoretical bounds \cite{grigory-blog}),
the $k$-machine model is characterized by one parameter (the number of machines) which allows one to develop and prove clean bounds
and serves as a basis for comparing various distributed algorithms. 
}

\onlyLong{
The $k$-machine model is also closely related to the classical \congest model~\cite{Peleg00}, and in
particular to the \emph{congested clique} model, which recently has received considerable attention (see,
e.g., \cite{LotkerPPP05,LenzenW11,Lenzen13,DruckerKO14,Nanongkai14,Censor-HillelKKLPS15,HegemanPPSS15,GhaffariP16,Jurdzinski018}). 
The main difference is that the $k$-machine model is aimed at the study of large-scale computations,
where the size $n$ of the input is significantly bigger than the number of available machines $k$, and thus
many vertices of the input graph are mapped to the same machine, whereas the two aforementioned models are aimed
at the study of distributed network algorithms, where $n = k$ and each vertex corresponds to a dedicated machine.
More ``local knowledge'' is
available per vertex (since it can access for free information about other vertices in the same machine) in the $k$-machine
model compared to the other two models. On the other hand, all vertices assigned to a machine have to communicate
through the links incident on this machine, which can limit the bandwidth (unlike the other two models where
each vertex has a dedicated processor). 
These differences manifest in the design of fast algorithms for these models. In particular, the best distributed algorithm in the congested clique model  may not directly yield the fastest algorithm in the $k$-machine model~\cite{PanduranganRS16}.
}

\onlyLong{
PageRank and triangle enumeration have received considerable attention in other models of distributed computing (see, e.g.,
\cite{SuriS11,ParkMK16,KoutrisBS16,Lin2,Lin3} and references therein). However, none of these results and techniques therein
can be translated to yield the bounds shown in this paper. A result that goes close is the lower bound for triangle
enumeration in the Massively Parallel Computation (MPC) model~\cite{KoutrisBS16}; this, however assumes
a worst-case initial partition of the input, whereas our lower bound holds even under a random (balanced) partition.
}

\onlyLong{
\subsection{Preliminaries}\label{sec:prelim}

\paragraph{\pagerank}
\pagerank is one of the most important measures to rank the importance
of nodes in a graph, and was first proposed
to rank  Web pages \cite{brin}.
The \pagerank of a graph $G = (V, E)$ is defined as follows. Let $\epsilon$ be a small constant which is fixed ($\epsilon$ is called the {\em reset} probability, i.e., with probability $\epsilon$ the random walk restarts from a node chosen uniformly at random among all nodes in the network). The \pagerank (vector) of a graph (e.g., see \cite{mcm-avrachenkov,ppr-bahmani2010,pr-survey+05,DasSarmaGP11}) is the {\em stationary distribution} vector $\pi$ of the following special type of random walk: at each step of the random walk, with probability $\epsilon$ the walk restarts from a randomly chosen node, and with probability $1-\epsilon$ the walk follows a randomly chosen outgoing (neighbor) edge from the current node and moves to that neighbor.
Computing the \pagerank and its variants efficiently in various computation models has been of tremendous research interest in both academia and industry. For a detailed survey of PageRank see, e.g., \cite{pr-survey+05,LangvilleM03}.  %

There are mainly two broad approaches to the \pagerank computation (see, e.g., \cite{bahmani11}).
One is the use of linear algebraic techniques (e.g., the Power Iteration \cite{page99}), and the other
is Monte Carlo methods~\cite{mcm-avrachenkov}.
In the Monte Carlo method, the basic idea is to approximate \pagerank by directly simulating the
corresponding random walk and then estimating the stationary distribution with the performed walk's distribution \cite{pagerank,mcm-avrachenkov}.

\paragraph{Triangle enumeration}
The triangle enumeration problem is to enumerate all the triangles in a graph, where a
triangle is a set of three vertices all adjacent to each other.\footnote{Sometimes this problem is also referred
to as \emph{triangle listing}, although there is a small difference: in triangle listing the output
must be generated and stored in memory, whereas in triangle enumeration the output is not required to
be stored. This distinction is relevant in bounded-memory models.}
This problem has attracted much interest because of its numerous practical applications,
including the analysis of social processes in networks~\cite{WattsS98,FoucaultVC10},
community detection~\cite{BerryHLP11}, dense subgraph mining~\cite{WangZTT10},
joins in databases~\cite{NgoRR13}, and the solution of systems of geometric
constraints~\cite{FudosH97}. The interested reader may refer to~\cite{ChuC12,BerryFNPSW15}
for additional applications.

Triangle detection and triangle counting are also well-studied problems, and potentially
significantly easier than triangle enumeration; however, we emphasize that for many
applications, including all the aforementioned ones, triangle detection and triangle
counting are not enough, and a complete enumeration of all the triangles is required.

The problem of finding triplets of vertices that consist of exactly two edges, usually called \emph{open triads},
has obvious applications, e.g., in social networks~\cite{wasserman}, where vertices
represent people, edges represent a friendship relation, and open triads can be used to recommend friends.
The problem of enumerating small subgraphs and cliques have numerous applications \cite{wasserman,WangZTT10,ChuC12,BerryFNPSW15}.
}

\onlyShort{\vspace{-0.1in}}
\section{Lower Bounds}
\label{sec:lb}
\onlyShort{\vspace{-0.05in}}
\subsection{A General Lower Bound Theorem}
\onlyShort{\vspace{-0.05in}}
In this section we present a result, called {\em General Lower Bound Theorem}, which provides a general way
to obtain round lower bounds in the $k$-machine model. In Section~\ref{sec:proofglb} we provide the full proof of
this result. We will then apply it to derive lower bounds for two graph problems, namely,
PageRank computation (Section~\ref{sec:prlower}) and triangle enumeration (Section~\ref{sec:trianglelb}).

Consider an $n$-vertex input graph $G$  partitioned across the machines via the random-vertex partition  in the $k$-machine model. Note that the input graph $G$ is sampled from a probability distribution on a (suitably chosen) set of graphs $\cG$. (For example, in the case of PageRank, $\cG$ is the set of all possible instantiations
of the lower bound graph $H$ shown in Figure \ref{fig:PRlowerbound}.) %
Consider a partition $\bp = (p_1,\dots,p_k)$ of an input graph $G$. We use boldface $\bp$ to denote a vector and $p_i$ to denote the $i$-th entry of $\bp$, which corresponds to the subgraph assigned to machine $M_i$.
In our analysis, we frequently condition on the event that a subgraph $p_i \subseteq G$ is assigned to a certain machine $M_i$. 
To simplify the notation, we also use $p_i$ to denote the event that this happens, e.g., $\Prob{ E \mid p_i }$ is the probability of event $E$ conditioned on the assignment of $p_i$ to machine $M_i$.

Let $\Pi_i$ be the random variable representing the transcript of the messages received by machine $M_i$ across its $k-1$ links when executing a given algorithm $\cA$ for (at most) $T$ rounds, and let $\prt$ be the set of all possible partitions of the graphs in $\cG$ among the $k$ machines. 
The execution of algorithm $\cA$ is fully determined by the given input partitioning $\bp \in \prt$ and the public random bit string $R \in \rand$, where $\rand$ is the set of all possible strings that are used as random bit string by the algorithm. 
Note that $R$ is itself a random variable.
Similarly as above, we write $\Prob{ E \mid p_i,r }$ when conditioning event $E$ on the events that the public random string is $r$ and machine $M_i$ obtains subgraph $p_i$ as its input, where $\bp = (p_1,\dots,p_i,\dots,p_k)$ and $(\bp,r) \in \prt\times\rand$. We use $\cA_i(\bp,r)$ to denote the output of machine $M_i$, when executing the algorithm for a given $(\bp,r)$. For technical reasons, we assume that the output $\cA_i(\bp,r)$ also includes $M_i$'s initial graph input $p_i$ and the random string $r$.\footnote{Any given algorithm can be modified to achieve this behavior by using only local computation.}
\begin{theorem}[General Lower Bound Theorem]\label{thm:expectedruntime}
Let $\infcost = \infcost(n,k)$ be a positive integer-valued function called \emph{information cost}, and let $Z$ be a random variable depending only on the input graph.
Consider a $T$-round $\epsilon$-error algorithm $\cA$, for some $\epsilon = o(\infcost/\entropy{Z})$, where $\entropy{Z}$ is the entropy of $Z$.
Let $\good \subseteq \prt\times\rand$ be a set of pairs $(\bp,r)$ where $\bp=(p_1,\dots,p_k) \in \prt$ is an input partition and $r \in \rand$ is a public random string, and $|\good| \ge (1-\epsilon - n^{-\Omega(1)})|\prt\times\rand|$.
Suppose that, for every $(\bp,r) \in \good$, there exists a machine $M_i$ receiving input graph $p_i$ and outputting $\cA_i(\bp,r)$, such that 
\begin{align} 
  \Prob{Z\!=z\! \mid p_i, r} &\le \left(\tfrac{1}{2}\right)^{\entropy{Z} - o(\infcost)},\label{eq:precond1} \\
  \Prob{Z\!=z\! \mid \cA_i(\bp,r),p_i,r} &\ge \left(\tfrac{1}{2}\right)^{\entropy{Z} - \infcost} \label{eq:precond2},
\end{align}
for every $z$ that has nonzero probability conditioned on events $\out_i\!=\!\cA_i(\bp,r)$, $P_i\!=\!p_i$, and $R\!=\!r$.
Then, if $B$ denotes the per-round communication link bandwidth, it holds that  
\begin{align} \label{eq:lb}
  T %
  =\Omega\left(\frac{\infcost}{B k}\right).
\end{align}
\end{theorem}

\paragraph{Intuition}
We can think of Premise~\eqref{eq:precond1} as bounding the initial knowledge of the machines about the random variable $Z$.
On the other hand, Premise~\eqref{eq:precond2} shows that at least one machine is able to increase its knowledge about the value of $Z$ eventually, which we formalize by conditioning on its output in addition to the initial knowledge. 
Then, if there is a large set (called $\good$) of inputs where these premises hold, then our theorem says that the worst-case time of the algorithm must be sufficiently large.
These insights are formally captured by the \emph{self-information} or \emph{surprisal} of an event $E$, which is defined as $\log_2(1/\Prob{E})$ \cite{reza1961introduction} and measures the ``amount of surprise'' or information contained in observing an event $E$. 
Premises~\eqref{eq:precond1} and \eqref{eq:precond2} imply that, from some machine $M_i$'s point of view, the occurrence of $\{Z\!=\!z\}$ is ``$\Omega(\infcost)$ more surprising'' given its initial knowledge, compared to observing this event after computing the output.
We can show that this surprisal change $\infcost$ bounds from below the maximum communication cost over all machines.
In this light, \eqref{eq:lb} tells us that the run time of the algorithm is roughly a $(1/k B)$-fraction of the maximum expected information cost.

\onlyShort{\vspace{-0.1in}}
\subsection{Proof of the General Lower Bound Theorem}
\onlyShort{\vspace{-0.05in}}
\label{sec:proofglb}
\onlyShort{In the proof of \Cref{thm:expectedruntime} we make use of some standard definitions in information theory (see e.g.,\cite{CoverT06}(also found in the full paper).}
\onlyLong{In the proof of \Cref{thm:expectedruntime} we make use of some standard definitions in information theory, which we now recall (and which can be found, e.g., in~\cite{CoverT06}).}
\onlyLong{
Consider random variables $X$, $Y$, and $W$. The 
\emph{entropy of $X$} is defined as $\entropy{X } = -\sum_{x} \Prob{X = x}\log_2\Prob{X=x}$, and the \emph{conditional entropy} is defined as
\begin{align} \label{eq:conditionalEntropy}
\entropy{X \mid Y} = \sum_{y} \Prob{Y\!=\!y}\ \entropy{X \mid Y=y}. 
\end{align}
The \emph{mutual information between $X$ and $Y$ given some event $\{W\!=\!w\}$} is denoted by $\information{ X; Y \mid W\!=\!w}$, and given by 
\begin{align}
\information{ X; Y \mid W\!=\!w} 
    &= \entropy{ X \mid W\!=\!w} - \entropy{ X \mid Y, W\!=\!w}. \label{eq:defMutual}
\end{align}
From this it immediately follows that
\begin{align}
  \entropy{X \mid W\!=\!w} \ge \information{X ; Y \mid W\!=\!w}. \label{eq:entropyMutual}
\end{align}
}

\paragraph{Critical Index} For a given input graph partition $\bp$ and a random string $r$, we are interested in identifying the machine that has the maximum expected value of the amount of information that its output reveals about the random variable $Z$.
This motivates us to define the \emph{critical index} function as
\begin{align} \label{eq:criticalIndex}
  \ell(\bp,r) := \argmax_{1 \le i \le k} \information{ \out_i ; Z \mid p_i, r },
\end{align}
and define random variables 
\onlyLong{
\begin{align}
  \text{$\Pi_*(\bp,r) = \Pi_{\ell(\bp,r)}(\bp,r)$ and
  $\out_*(\bp,r) = \out_{\ell(\bp,r)}(\bp,r)$}. \label{eq:randomvars}
\end{align}
}
Intuitively speaking, for each $(\bp,r)\in\prt\times\rand$, the random variable $\out_*$ is the output of the machine $M_i$ (where $i$ depends on $\bp,r$) that attains the maximum 
mutual information between its output and the random variable $Z$.
For a given $(\bp,r)$, we use 
\begin{align} \label{eq:pstar}
p_* = p_{\ell(\bp,r)}
\end{align}
to denote the input partition of machine $M_{\ell(\bp,r)}$.
Note that $Z$ depends \emph{only} on the input graph, whereas $\Pi_*$, $P_*$, and $\out_*$ depend on the input graph and, in addition, also on the chosen partition $\bp$ and random string $r$.
From \eqref{eq:criticalIndex}, we immediately obtain the following property of the critical index.
\begin{observation} \label{obs:mutual}
  For all $(\bp,r) \in \prt\times\rand$, and for all $i \in [k]$, it holds that
  \[
  \information{\out_* ; Z \mid p_*, r} \ge \information{\out_i ; Z \mid p_i, r},
  \]
  where $p_* = p_{\ell(\bp,r)}$ and $\bp = (p_1,\dots,p_{\ell(\bp,r)},\dots,p_k)$.
\end{observation}

\begin{lemma} \label{lem:cc}
For every $(\bp,r) \in \prt\times\rand$ where $\bp = (p_1,\dots,p_*,\dots,p_k)$, it holds that
\[
\information{ \Pi_* ; Z \mid p_*, r} \ge \information{ \out_* ; Z \mid p_*, r}. 
\]
\end{lemma}
\onlyLong{
\begin{proof}
Consider a $(\bp,r)\in\prt\times\rand$ as described in the premise of the lemma.
It holds that
\begin{align*}
  \information{ \Pi_* ; Z \mid p_*, r}
    &\ge\max_{1 \le i \le k} \information{ \Pi_i ; Z \mid p_i, r} \tag{by Obs.~\ref{obs:mutual}}\\
    &=\max_{1 \le i \le k} \left( \entropy{ Z \mid p_i, r } - \entropy{ Z \mid \Pi_i, p_i, r }\right).\tag{by \eqref{eq:defMutual}}
\end{align*}
The random variable $\out_i$ which represents the output of machine $M_i$ is fully determined by the transcript $\Pi_i$, $M_i$'s input graph assignment (i.e., the random variable $P_i$), and the random bits $R$.
Therefore, we can use the bound
$\entropy{Z \mid \Pi_i, p_i,r} \le \entropy{Z \mid \out_i,p_i,r}$ in the right-hand side of the above inequality to obtain 
\begin{align*}
\information{ \Pi_* ; Z \mid p_*, r}
&\ge\max_{1 \le i \le k} \left( \entropy{ Z \mid p_i, r } - \entropy{ Z \mid \out_i, p_i, r }\right)\notag\\
   &= \max_{1 \le i \le k} \information{ \out_i ; Z \mid p_i, r} \notag\\
   &=\information{ \out_{\ell(\bp,r)} ; Z \mid p_{\ell(\bp,r)}, r} \tag{by definition of critical index, cf.\ \eqref{eq:criticalIndex}}\\
   &=\information{ \out_* ; Z \mid p_*, r},  \tag{by \eqref{eq:randomvars} and \eqref{eq:pstar}} 
\end{align*}
and the lemma follows.
\end{proof}
}
\begin{lemma} \label{lem:mutual}
  For all $(\bp,r) \in \good$ where $\bp=(p_1,\dots,p_k)$, there is an $i \in [k]$ (which satisfies (1) and (2) in the premise of the theorem) such that 
$    \information{\out_i ; Z \mid p_i, r} \ge \infcost - o(\infcost).$
\end{lemma}
\begin{proof}
  For a given $(\bp,r)\in\good$, let $M_i$ be a machine satisfying 
  \eqref{eq:precond2} (in addition to \eqref{eq:precond1}).
  By definition, %
  \begin{align}
    \label{eq:mutual}
    \information{\out_i ; Z \mid p_i, r} =\entropy{ Z \mid p_i, r} - \entropy{ Z \mid \out_i, p_i, r}.
  \end{align}
  We will now bound the terms on the right-hand side.
  By definition, we obtain
  \begin{align}
  \entropy{Z \mid p_i, r} &= - \sum_z \Prob{Z=z \mid p_i, r } \log_2\Prob{Z=z \mid p_i, r }\notag \\
  &\ge \left(\entropy{Z} - o(\infcost)\right)\sum_z \Prob{Z=z \mid p_i, r }\tag{by \eqref{eq:precond1}} \\
 & = \entropy{Z} - o(\infcost), \label{eq:condEntLow}
\end{align}
where the last inequality follows from $\sum_z \Prob{Z=z \mid p_i, r }= 1$. 

In the remainder of the proof, we derive an upper bound on $\entropy{ Z \mid \out_i, p_i, r}$.
Since 
\begin{align} 
  \entropy{ Z \mid \out_i, p_i, r} \le \entropy{Z \mid \out_i}, \label{eq:entUp2}
\end{align} 
we will proceed by proving an upper bound on the latter term.
To simplify the notation, we use ``$\cA_i(\bp,r)$'' as a shorthand for the event ``$\out_i = \cA_i(\bp,r)$''.  
By definition, we have
\begin{align} 
  \entropy{ Z \mid \out_i } 
  &= \sum_{\substack{(\bp,r) }} \Prob{\cA_i(\bp,r)}\ \entropy{ Z \mid \cA_i(\bp,r)}  \notag\\
  &= \!\!\!\sum_{\substack{(\bp,r) \in \good}} \Prob{\cA_i(\bp,r)}\ \entropy{ Z \mid \cA_i(\bp,r)} + \sum_{\substack{(\bp,r) \notin \good}} \Prob{\cA_i(\bp,r)}\ \entropy{ Z \mid \cA_i(\bp,r)}\notag \\
  &\le\!\!\! \sum_{\substack{(\bp,r) \in \good}} \Prob{\cA_i(\bp,r)}\ \entropy{ Z \mid \cA_i(\bp,r)} + \entropy{ Z }\left(\sum_{\substack{(\bp,r) \notin \good}} \Prob{\cA_i(\bp,r)}\right), \label{eq:entterm}
\end{align}
where the last inequality follows from $\entropy{Z} \ge \entropy{Z \mid \cA_i(\bp,r)}$.
Intuitively speaking, the first sum in \eqref{eq:entterm} represents the remaining uncertainty of $Z$ upon termination, assuming machines start with a hard input assignment (i.e., in $\good$), whereas the second term is weighted by the probability that either the input was easy or the algorithm failed (i.e. $\notin \good$).
The following claim bounds the entropy term in the first sum of \eqref{eq:entterm}, where $(\bp,r)$ is restricted to the set $\good$. 
\begin{claim} \label{cl:upper1}
  $\entropy{ Z \mid \cA_i(\bp,r)} \le  \entropy{Z} - \infcost$.
\end{claim}

\begin{proof}[Proof of Claim~\ref{cl:upper1}]
From the definition of entropy, we obtain
\begin{align}
  \entropy{ Z \mid \cA_i(\bp,r)} &=  - \sum_z \Prob{Z=z \mid \cA_i(\bp,r)}  \log_2\Prob{Z\!=\!z \mid \cA_i(\bp,r)}. \label{eq:entterm2}
\end{align}
Since we assume that machine $M_i$ also outputs its initial graph assignment (i.e., $p_i$) and the public random string $r$, it holds that
\[
\entropy{ Z \mid \cA_i(\bp,r)} = \entropy{ Z \mid \cA_i(\bp,r), p_i, r},  \notag
\]
which allows us to rewrite \eqref{eq:entterm2} as
\[
\entropy{ Z \mid \cA_i(\bp,r)} =  - \sum_z \Prob{Z=z \mid \cA_i(\bp,r), p_i, r} \cdot \log_2\Prob{Z\!=\!z \mid \cA_i(\bp,r), p_i, r}.
\]
Recalling that $M_i$ satisfies \eqref{eq:precond2}, we get
\begin{align*}
  \entropy{ Z \mid \cA_i(\bp,r)} 
  & \le \left(\entropy{Z} - \infcost\right) \sum_z \Prob{Z=z \mid \cA_i(\bp,r), p_i, r} = \entropy{Z} - \infcost, \notag
\end{align*}
since $\sum_z \Prob{Z=z \mid \cA_i(\bp,r), p_i, r} = 1$.
\end{proof}

We will now derive an upper bound on the second sum in \eqref{eq:entterm}.
\begin{claim}\label{cl:upper2}
  $\sum_{\substack{(\bp,r) \notin \good}} \text{\rm Pr}[\cA_i(\bp,r)]
  \le \epsilon + n^{-\Omega(1)}$.
\end{claim}
\begin{proof}[Proof of Claim~\ref{cl:upper2}]
  Consider the set $(\prt\times\rand) \setminus \good$.
  According to our model, the input graph and its partitioning among the machines correspond to choosing, uniformly at random, an element from $\prt$, whereas the random string $r$ is uniformly selected from $\rand$. 
  Since the output of machine $M_i$ is fully determined by $(\bp,r)$, we have
  \[
    \sum_{\substack{(\bp,r) \notin \good}}\!\!\!\!\! \Prob{\cA_i(\bp,r)} 
    = \!\!\!\!\!\! \sum_{\substack{(\bp,r) \notin \good}}\!\!\!\!\! \Prob{(\bp,r)}
    = \Prob{ (\prt\times\rand) \setminus \good }.
  \]
From the lower bound on the size of $\good$ in the theorem premise, we obtain an upper bound such that
  \begin{align}
  \sum_{\substack{(\bp,r) \notin \good}}\!\!\!\!\! \Prob{\cA_i(\bp,r)} 
  = \Prob{(\prt\times\rand)\setminus\good} \le \epsilon + n^{-\Omega(1)},  \notag
  \end{align}
thus proving the claim.  
\end{proof}
  Plugging the bounds in Claims~\ref{cl:upper1} and \ref{cl:upper2} into  \eqref{eq:entterm}, we get
  \begin{align*}
     \entropy{ Z \mid \out_i} 
    &\le \left(\entropy{Z} - \infcost\right)\sum_{\substack{(\bp,r) \in \good}} \Prob{\cA_i(\bp,r)} + \entropy{ Z } \left(\epsilon + n^{-\Omega(1)} \right) \\
    &\le \left(\entropy{Z} - \infcost\right)
    + \entropy{ Z } \left(\epsilon + n^{-\Omega(1)} \right).
 \end{align*}     
Assuming a sufficiently large constant in the exponent of $n^{-\Omega(1)}$, 
we observe that $\entropy{Z} \cdot n^{-\Omega(1)} = o(1)$ since $Z$ depends only on the input graph. 
By the premise of Theorem~\ref{thm:expectedruntime}, we have $\epsilon = o(\infcost/\entropy{Z})$
and $\infcost \le \entropy{Z}$, hence $\epsilon \cdot\entropy{Z} = o(\infcost)$.
From this and \eqref{eq:entUp2} we conclude that
\[
\entropy{ Z \mid \out_i, p_i, r} \leq \entropy{Z} - \infcost + o(\infcost).
\]
Plugging this upper bound and the lower bound of \eqref{eq:condEntLow} into the right-hand side of \eqref{eq:mutual}, completes the proof of Lemma~\ref{lem:mutual}.
\end{proof}

Recall that Lemma~\ref{lem:cc} holds for any $(\bp,r) \in \prt\times\rand$; in particular, even if we restrict our choice to the set $\good$.
Thus, for $(\bp,r) \in \good$, where $\bp = (p_1,\dots,p_k)$, let $i \in [k]$ be the index for which Lemma~\ref{lem:mutual} holds (which is the index of the machine satisfying Premises~\eqref{eq:precond1} and \eqref{eq:precond2}).
\onlyLong{%
This yields
\begin{align} %
  \entropy{\Pi_* \mid p_*, r}
  &\ge
  \information{ \Pi_* ; Z \mid p_*, r} \tag{by \eqref{eq:entropyMutual} }\\
  &\ge
  \information{\out_* ; Z \mid p_*, r} 
  \tag{by Lemma~\ref{lem:cc}}  \notag\\
  &\ge \information{\out_i ; Z \mid p_i, r} \tag{by Obs.~\ref{obs:mutual}}  \notag\\
  &\ge  \infcost - o(\infcost),      \label{eq:transcriptLower}
\end{align}
}
where the last inequality follows from Lemma~\ref{lem:mutual}. %
To complete the proof of Theorem~\ref{thm:expectedruntime}, we will argue that the worst-case run time needs to be large, as otherwise the entropy of machine $M_{\ell(\bp,r)}$'s transcript $\Pi_*$ would be less than $\infcost - o(\infcost)$.
The value of $\entropy{\Pi_* \mid p_*, r}$ is maximized if the distribution of $(\Pi_* \mid p_*, r)$ is uniform over all possible choices. In the next lemma we show that, during $T$ rounds of the algorithm, the transcript can take at most $2^{(B+1)(k-1)T}$ distinct values, and thus 
\begin{align}
  \entropy{\Pi_* \mid p*,r } \le \log_2\left(2^{(B+1)(k-1)T}\right) = O(B\ k\ T). \label{eq:transcriptUpper}
\end{align}

\begin{lemma}\label{lem:encoding}
  Suppose that some machine $M_i$ can receive a message of at most $B$ bits on each of its $k-1$ links in a single round.
  Let $\Gamma$ be the bits received by $M_i$ over its $k-1$ links during $T$ rounds. 
  Then, $\Gamma$ can take at most $2^{(k-1)(B+1)T}$ distinct values.  
\end{lemma}
\onlyLong{
\begin{proof}
Since in a synchronous model one can convey information even by \emph{not} sending any bits in a given round,
there are at most $2^B + 1 < 2^{B+1}$ distinct possibilities for the communication received over a single link of bandwidth $B$ in any given round.
Thus, we can view the communication received over $M_i$'s $k-1$ links as a word $\omega_1$ of length $k-1$, where each character of $\omega_1$ is chosen from an alphabet of size (at most) $2^{B+1}$, resulting in $2^{(B+1)(k-1)}$ possible choices for $\omega_1$.
Finally, we view $\Gamma$, i.e., the communication received over the $T$ rounds, as a word of length $T$, where the alphabet size of each character is $\le 2^{(B+1)(k-1)}$, yielding $2^{(B+1)(k-1)T}$ many choices in total.
\end{proof}
}
Recall that the run time $T$ is the maximum time required by any machine $M_i$, over all random strings and input assignments, i.e., $T = \max_{(\bp,r)}\ T(\bp,r)$.
Combining \eqref{eq:transcriptLower} and \eqref{eq:transcriptUpper}, it follows that 
\begin{align*} %
  T  = \max_{(\bp,r)}\ T(\bp,r) 
   = \Omega\left( \frac{IC}{B k}\right).
\end{align*}
This completes the proof of Theorem~\ref{thm:expectedruntime}.

\onlyShort{\vspace{-0.1in}}
\subsection{A Lower Bound for PageRank Computation}\label{sec:prlower}
\onlyShort{\vspace{-0.05in}}

\begin{theorem}\label{thm:pageranklb}
Let $\cA$ be an algorithm that computes a $\delta$-approximation of the $\pagerank$ vector of an $n$-node graph for a small constant $\delta>0$ (depending on the reset probability), and suppose that $\cA$ succeeds with probability $\ge 1 - o(1/k)$.
Then, the run time of $\cA$ is $\Omega\left(\frac{n}{B \cdot k^2}\right)$, assuming a communication link bandwidth of $B$ bits per round and $k=\Omega(\log^2 n)$ machines.
This holds even when the input graph is assigned to the machines via random vertex partitioning. 
\end{theorem}

We first give a high-level overview of the proof. 
As input graph $G$, we construct a weakly connected directed graph where the direction of certain ``important'' edges
is determined by a random bit vector, and assign random IDs to all the vertices. 
Flipping the direction of an important edge changes the \pagerank of connected vertices by a constant factor and hence any (correct) algorithm needs to know about these edge directions.
It is crucial that the vertex IDs are chosen randomly, to ensure that knowing just the direction of important edges is not sufficient for computing the \pagerank of the adjacent nodes, as these random vertex IDs ``obfuscate the position'' of a vertex in the graph. 
This means that a machine needs to know both, the direction of an important edge and the IDs of the connected vertices to be able to output a correct result.
By using a Chernoff bound, we can show that the random vertex partitioning of the input graph does not reveal too many edge-directions together with the matching vertex IDs to a single machine.
This sets the stage for applying our generic lower bound theorem (\Cref{thm:expectedruntime}) to obtain a lower bound on the run time.

\noindent\textbf{The Lower Bound Graph.}
We consider the following directed graph $H$ (see Figure~\ref{fig:PRlowerbound}) of $n$ vertices and $m=n-1$ edges; for simplicity, assume that $m/4$ is an integer. 
Let $X=\{x_1,x_2,\dots,x_{m/4}\}$, $U=\{u_1,u_2,\dots,u_{m/4}\}$, $T=\{t_1,t_2,\dots,t_{m/4}\}$, $V=\{v_1,v_2,\dots,v_{m/4}\}$,
and let $V(G) = \{ X \cup U \cup T \cup V \cup \{w\} \}$.
The edges between these vertices are given as follows:
For $1 \le i \le {m/4}$, there is a directed edge $u_i \ra t_i$, a directed edge $t_i \ra v_i$, and a directed edge $v_i \ra w$. 
The edges between $u_i$ and $x_i$ (these are the ``important'' edges mentioned above) are determined by a bit vector $\mathbf{b}$ of length $m/4$ where each entry $b_i$ of $\mathbf{b}$ is determined by a fair coin flip:
If $b_i = 0$ then there is an edge $u_i \ra x_i$, otherwise there is an edge $x_i \ra u_i$.
Lemma~\ref{lem:sep} \onlyShort{(in full paper)} shows that, for any $1 \le i \le m/4$ and for any $\epsilon < 1$, there is a constant factor separation
between the \pagerank of any node $v_i$
if we switch the direction of the edge between $x_i$ and $u_i$.

\begin{figure}[h]
    \begin{center}
 \includegraphics[width=0.65\textwidth]{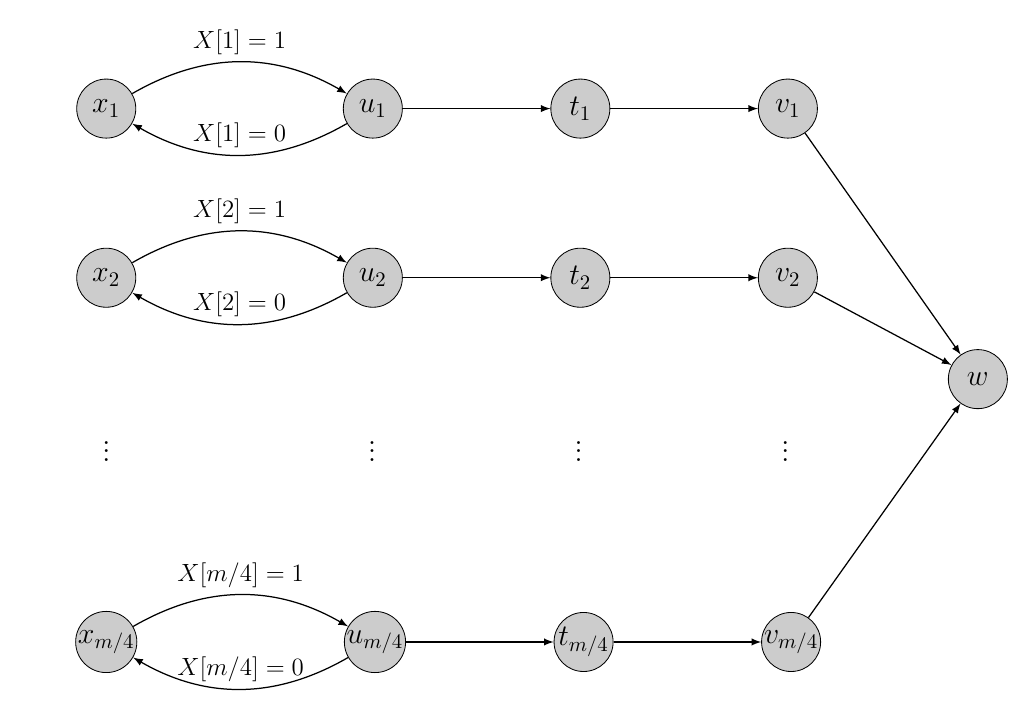}
       \caption{The graph $H$ used to derive a lower bound on the round complexity of PageRank computations.}
       \label{fig:PRlowerbound}
\vspace{-0.25in}
\label{fig:network-ex}
   \end{center}
\end{figure}

\begin{lemma} \label{lem:sep}
The following holds for the $\pagerank$ value of vertices $v_i$ of $G$, for $1 \le i \le {n/4}$:
If $b_i = 0$, then $\pagerank(v_i) = \frac{(2.5 - 2\epsilon + \epsilon^2 /2)\epsilon}{n}$. 
Otherwise, if $b_i = 1$, then $\pagerank(v_i) \geq \frac{(3 - 3\epsilon + \epsilon^2)\epsilon}{n}$.
For any $\epsilon < 1$, there is a constant factor (where the constant depends on $\epsilon$) separation between the two cases.
\end{lemma}
\onlyLong{

\begin{proof}
We will determine an estimate of $\pagerank(v_i)$ using the distributed random walk approach described
at the beginning of Section~\ref{sec:pralgo}, whereby the expected number of random walk tokens addressed to one node, multiplied by
$\epsilon/cn\log n$, gives a high-probability estimate of the PageRank value of the node. The expected value of $\psi_{v_i}$ is
\[
\expect{\psi_{v_i} | b_i = 0}  = c \log n \mleft(1 + (1-\epsilon) + \frac{(1-\epsilon)^2}{2}\mright)
\]
and
\[
\expect{\psi_{v_i} | b_i = 1} = c \log n \mleft(1 + (1-\epsilon) + (1-\epsilon)^2 + (1-\epsilon)^3\mright).
\]
Therefore,
\[
\pagerank(v_i) = \frac{(2.5 - 2\epsilon + \epsilon^2 /2)\epsilon}{n}
\]
if $b_i = 0$, and
\[
\pagerank(v_i) \geq \frac{(3 - 3\epsilon + \epsilon^2)\epsilon}{n}
\]
 if $b_i = 1$.
\end{proof}
}

\noindent\textbf{The Input Graph Distribution.} We now build our input graph $G$ as follows.
Let $m=n-1$, and let $ID$ be the random variable representing a set of $n$ unique integers chosen uniformly at random
from $\{ S \subset [1,\poly(n)]\colon |S|=n \}$.
Assigning each vertex of $H$ a unique integer from $ID$ (in an arbitrary predetermined way) yields a graph $G$.
Let $\cal G$ denote the set of graphs $G$ determined by all possible (different) ID assignments to all possible instances of $H$ considering all possible edge directions.
Let $\prt$ be the set of all input graph partitions (i.e., the set of all graphs in $\cal G$ and all their possible input partitions) among the $k$ machines, and let $\rand$ be the set of all random strings used by a given $\pagerank$ algorithm $\cA$.  
Let $\bal \subseteq \prt$ be the set of all input partitions where each machine receives $\tilde\Theta(n/k)$ vertices of the input graph.
Note that $(\bp,r) \in \prt\times\rand$ fully determines the run of $\cA$.
We assume that each machine $M_i$ outputs a set $\{(\pi_1,id_1),\dots,(\pi_\ell,id_{\ell})\}$,
where $\pi_j$ refers to the \pagerank value of the vertex with ID $id_j$.
Note that we do not make assumptions neither on which machine being the one that outputs the \pagerank of a specific vertex
$v$ (which could be a machine that holds no initial knowledge about $v$ and its ID), nor on the individual sizes of these output sets.

\noindent\textbf{Discovering Weakly Connected Paths of Vertices.}
By the random vertex partitioning, each machine $M_i$ initially holds $\tilde\Theta(n/k)$ vertices in total.
More specifically, $M_i$ receives random sets $X_i \subseteq X$, $U_i \subseteq U$, $T_i \subseteq T$, and $V_i \subseteq V$, each containing $O(n\log (n)/k)$ vertices.
As machine $M_i$ also gets to know the incident edges of these vertices, $M_i$ can locally check if a path induced by some $(x_{j_1},u_{j_2},t_{j_3},v_{j_4}) \in X_i \times U_i \times T_i \times V_i$ is weakly connected, i.e., $j_1 = \cdots = j_4$.
Since $M_i$ learns the output pair $(\pagerank(v),id_v)$ at zero cost, we upper bound the number of such paths that the machines learn initially by using a Chernoff bound.%
\onlyShort{ That is, we show that each machine learns at most $O\left(\frac{n\log n}{k^{3/2}}\right)$ paths.}
\begin{lemma} \label{lem:matching}
With probability at least $1 - n^{-4}$, the initial graph partition reveals at most $O\left(\frac{n\log n}{k^{2}}\right)$ weakly connected paths between vertices in $X$ and $V$ %
to every machine. 
\end{lemma}

\begin{proof}
Fix one machine $M_i$. %
If a vertex is assigned to $M_i$, then machine $M_i$ knows its incident edges and the IDs of their endpoints.
Therefore, $M_i$ can discover a weakly connected path (between $X$ and $V$) in one of the following ways:
(1) $M_i$ obtains $x_{j} \in X$ and $t_{j} \in T$;
(2) $M_i$ obtains $u_{j} \in U$ and $v_{j} \in V$.
The argument is similar in both cases and hence we focus on (1) for the rest of this proof.
By the random vertex partition process, the probability that $x_j$ and $t_{j}$ both are assigned to machine $M_i$ is $\frac{1}{k^2}$.
Since all vertices are assigned independently at random, a standard Chernoff bound shows that with high probability $O(n \log n / k^2)$ matching vertex pairs $(x_j,t_j)$ are assigned to machine $M_i$.
Applying the union bound over the $k$ machines completes the proof.
\end{proof}

\noindent\textbf{Good Inputs.} We define $\good \subseteq \bal\times\rand$ to be the set of all (balanced) inputs and random strings where 
(1) $\cA$ correctly outputs the $\pagerank$ of each vertex, (2) partition $\textbf{p}$ is ``balanced'', i.e., each machine is assigned
$O(n\log n /k)$ vertices (and hence $O(n \log n /k)$ edges since $m=O(n)$), and (3) the partitioning is such that each machine
knows at most $O((n \log n)/k^{2})$ weakly connected paths initially; we define $\bad = \prt\times\rand\setminus \good$.

\begin{lemma} \label{lem:goodPageRank}
(A) For any $(\bp,r) \in \good$, algorithm $\cA$ is correct and there must be at least one machine $M_i$ whose output list contains $\Omega(n/k)$ vertices of $V$.
(B) $|\good| \ge \left(1 - o(1/k) - n^{-\Omega(1)}\right)|\prt\times\rand|.$
\end{lemma}
\onlyLong{
\begin{proof}
Part~(A) follows directly from the definition of set $\good$.
For (B), note that $\cA$ succeeds with probability $\ge 1 - o(1/k)$.
Moreover, the random vertex partitioning ensures that each machine receives $\tilde \Theta(n\log(n)/k)$ vertices with probability $\ge 1 - n^{-4}$.
Hence, the above is true for at least a $\left(1 - o(1/k) - n^{-4}\right)$-fraction of the possible graph partition and random string pairs in $\prt\times\rand$.
\end{proof}
}

To instantiate \Cref{thm:expectedruntime}, we show in \Cref{lem:pagerankEntropyLB} and \Cref{lem:pagerankEntropyUB} that we can satisfy the Premises~\eqref{eq:precond1} and \eqref{eq:precond2}, by setting $\infcost = m/4k = \Theta(n/k)$.  Plugging the above value of $\infcost$ in \eqref{eq:lb} then gives the claimed lower bound.

\begin{lemma} \label{lem:pagerankEntropyLB}
  Let $Z$ be the random variable representing the set of pairs $\{(b_1,v_1),\dots,(b_{m/4},v_{m/4})\}$, where $b_j$ refers to the direction of the edge $(x_j,u_j)$ in the weakly connected path $(x_j,u_j,t_j,v_j)$ of the input graph of \Cref{fig:PRlowerbound}.
  Then, for each $(\bp,r) \in \good$, where $\bp=(p_1,\dots,p_k)$, and for every possible choice of $z$, it holds that
  $\Prob{ Z = z \mid p_i, r } \le 2^{-\left({m}/{4} - O(n \log(n)/k^2)\right)}$.
\end{lemma}
\onlyLong{
\begin{proof}
Consider a $(\bp,r) \in \good$ where $\bp = (p_1,\dots,p_i,\dots,p_k)$. By \Cref{lem:goodPageRank}(A), that algorithm $\cA$ correctly computes the PageRank and some machine (without loss of generality) $M_i$ outputs at least $\Omega(n/k)$ PageRank values. 

By \Cref{lem:sep}, we know that algorithm $\cA$ can only correctly output $\pagerank(v_j)$ at machine $M_i$ if $M_i$ knows the direction of the edge between $u_j$ and $x_j$ (from \Cref{lem:sep}, since the direction of the corresponding edge can be derived from the $\pagerank$ value).
This means that if machine $M_i$ outputs the \pagerank\ for $v_j$ as a pair $(\pi_j,v_j)$, then it can reconstruct the pair $(b_j,v_j)$, for any $1 \le j \le m/4$.

Since $(\bp,r) \in \good$, it follows by \Cref{lem:matching} that each machine $M_i$ learns at most $\eta = O(n\log (n)/k^{2})$ output entries of $V$ for free by inspecting its assigned input.
In addition to these $\eta$ entries, $M_i$ might know partial information about the remaining $\Omega(n)-\eta$ pairs. 

It follows that, for each of the other weakly connected paths that are not concerned with its $\eta$ already known $\pagerank$ values, $M_i$ either has initial knowledge of the index $\ell$ of the respective vertex $v_\ell \in V_i$, or it knows the edge direction $b_\ell$ between $x_\ell$ and $u_\ell$, but not both. 
Notice that knowledge of the vertex ID of $v_\ell$ reveals no additional information about the index $\ell$ since we choose vertex IDs uniformly at random.
We refer to these paths as being \emph{partially known to $M_i$}.

It follows that, for each index $j$ for which the path is partially known to $M_i$, there are two possibilities $(0,v_j)$ and $(1,v_j)$, each of which is equally likely, according to the input distribution.

Therefore, taking into account the initial input assignment, we still have at least $2^{m/4 - O(n\log(n)/k^{2})}$ possible choices for $z$, i.e., the output of $M_i$ concerning vertices in $V$, each of which is equally likely without conditioning on further knowledge.
Thus, 
\[
 \Prob{Z = z \mid  p_i,r} \le 2^{-(m/4 - O(n\log (n)/k^{2}))},
\]
completing the proof of the lemma.
\end{proof}
}

\begin{lemma} \label{lem:pagerankEntropyUB}
For each $(\bp,r) \in \good$, where $\bp=(p_1,\dots,p_k)$, there exists a machine $M_i$ 
with output $\cA_i(\bp,r)$ such that, for every choice of $z$ for $Z$ (defined in \Cref{lem:pagerankEntropyLB}) that has nonzero probability conditioned on $\cA_i(\bp,r),p_i,r$, it holds that
$
  \Prob{ Z = z \mid \cA_i(\bp,r),p_i, r } \ge 1/2^{\frac{m}{4} - \frac{m}{4k}}.
$
\end{lemma}
\onlyLong{
\begin{proof}
By \Cref{lem:goodPageRank}, we know that there is a machine $M_i$ that outputs at least $m/4k$ $\pagerank$ values of vertices in $V$. 
Let $\lambda$ be the total number of pairs $(b_j,v_j)$, where $b_j$ is the direction of the edge $(x_j,u_j)$ in the weakly connected path $(x_j,u_j,t_j,v_j)$ (cf.\ \Cref{lem:pagerankEntropyLB}) that remain unknown to machine $M_i$ conditioned on its input $p_i$, random string $r$, and its output $o_i$.

Observing that the size of its output $o_i$ is $\ge m/4k$ and the fact that we can recover the pair $(b_j,v_j)$ if $M_i$ outputs the \pagerank\ of $v_j$ (see proof of \Cref{lem:pagerankEntropyLB}), it follows that $\lambda \le \frac{m}{4} - \frac{m}{4k}$, and thus there are $2^{\frac{m}{4} - \frac{m}{4k}}$ distinct choices for $z$.
The probability bound is minimized if each remaining possible choices of $z$ are equally likely. 
This implies that $\Prob{Z \mid o_i,p_i,r} \ge 1/2^{\frac{m}{4} - \frac{m}{4k}}$, as desired.
\end{proof}
}

\onlyShort{\vspace{-0.1in}}
\subsection{A Lower Bound for Triangle Enumeration}\label{sec:trianglelb}
\onlyShort{\vspace{-0.05in}}

We first give a high-level overview of the proof.
The input graphs that we use for our lower bounds are sampled according to the $G_{n,1/2}$ Erd\"os-Renyi random graph model.
We will argue that enumerating triangles implies a large reduction of the entropy of the characteristic vector of edges $Z$, i.e.,
$Z$ is a bit vector whose entries reflect the presence/absence of an edge in the input graph. 
We prove that initially the machines do not have significant knowledge of $Z$, which is equivalent to having a small probability for the event $\{Z=z\}$, for any $z$.
Then, we show that any machine that outputs $t/k$ triangles, for a parameter $t$, must have reduced its
uncertainty about $Z$ by approximately $(t/k)^{2/3}$ bits. 
In other words, the information obtained by such a machine throughout the course of the algorithm is high.
We  apply \Cref{thm:expectedruntime} to obtain a lower bound on the run time of any algorithm. This yields the following result.

\begin{theorem} \label{thm:trianglelb}
There exists a class of graphs $\cG$ of $n$ nodes for which every distributed algorithm that solves triangle enumeration
in the $k$-machine model has a time complexity of 
$\Omega\left(\frac{n^{2}}{B \cdot k^{5/3}}\right)$, assuming a link bandwidth of $B$ bits per round, $k=\Omega(\log n)$ machines, and an error probability of $\epsilon = o(k^{-2/3})$.
This holds even when the input graph is assigned to the machines via random vertex partitioning. 
\end{theorem}

\paragraph{The Input Graph Distribution}
We choose our input graphs according to the Erd\"os-Renyi random graph model $G_{n,1/2}$,
which samples an $n$-node graph where each possible edge is included independently with probability $1/2$.
We use $\prt$ to denote the set of all possible partitions of all possible sampled $n$-node graphs and, similarly to before, denote the set of all random strings used by the algorithm by $\rand$.

Let $Z$ be the characteristic vector of the edges\footnote{The characteristic vector specifies the graph $G$. Order the ${n \choose 2}$ possible edges in some fixed ordering; if the $j$th edge in this ordering appears
in $G$, then $Z_j = 1$, otherwise it is 0.} of the input graph $G$.
Note that the execution of $\cA$ is fully determined by the given graph input partition $\textbf{p}=(p_1,\dots,p_k) \in \prt$ and the shared (among all machines) random bit string $r \in \rand$, where $\rand$ is the set of all possible strings that are used as random bit string by the algorithm.
Hence we have $|\prt \times \rand|$ possible outcomes when running $\cA$ on a graph sampled from $\cG$.

\paragraph{Good Inputs} \label{def:goodinputs} We define $\good \subseteq \prt \times \rand $ to be the set of input pairs $(\bp,r)$ such that (1) $\cA$ performs correctly for the graph partition $\textbf{p}$ of graph $G$ and the random string $r$, (2) partition $\textbf{p}$ is ``balanced'', i.e., each machine is assigned $O(n\log(n) /k)$ vertices (and hence $O(n^2 \log(n) /k)$ edges), and (3) $G$ has $\ge t$ triangles, for some fixed $t = \Theta({n \choose 3})$.

\onlyLong{
\begin{lemma}[Good Inputs] \label{lem:good}
(A) For every $(\bp,r) \in \good$, at least one machine outputs $\ge t/k$ triangles when executing algorithm $\cA$ with $(\bp,r)$, and
(B) $|\good| \ge (1-\epsilon')|\prt\times\rand|$, where $\epsilon' = \epsilon - n^{-\Omega(1)}$.
\end{lemma}

\begin{proof}
Part (A) is immediate from the definition of $\good$.
For (B), note that $\cA$ succeeds with probability $\ge 1 - \epsilon$ and the random vertex partitioning guarantees a balanced partition with probability $\ge 1 - n^{-4}$.
We know from Equation~4.10 in \cite{janson}, that the number of triangles in a input graph $G$ sampled from $G_{n,1/2}$ is
$\Theta({n \choose 3})$ with probability $\ge 1 - e^{-\Omega(1)}$, and hence the size of $\good$ contains all except
at most a $(1 -  \epsilon - n^{-3})$-fraction of the graphs in $\prt\times\rand$.
\end{proof}
}

\onlyShort{
Similarly as for \Cref{thm:pageranklb}, we prove that
$\Prob{Z\!=\!z \mid p_i,r} \le 1/2^{{n \choose 2} - O\left({n^2\log n}/{k}\right)}$ and
$\Prob{Z=z \mid o_i, p_i,r } \ge 1/2^{{n \choose 2} - \frac{1}{3}\left(\frac{t}{k} - O\left(\frac{n^2\log n}{k}\right)\right)^{2/3}}$.
Then, setting $\infcost = \Theta(n^2/k^{2/3})$,
and applying \Cref{thm:expectedruntime}, completes the proof of \Cref{thm:trianglelb}.
}%
\onlyLong{%
\begin{lemma} \label{lem:entropyLB1}
Let random variable $Z$ denote the characteristic vector of the edges of the sampled input graph $G$. 
For every $(\bp,r) \in \good$ where $\bp=(p_1,\dots,p_k)$ and every characteristic edge vector $z$, it holds that
$
\Prob{Z\!=\!z \mid p_i,r} \le 1/2^{{n \choose 2} - O\left({n^2\log(n)}/{k}\right)},
$ for every $i \in [1,k]$. 
\end{lemma}

\begin{proof}
For any $(\bp,r) \in \good$, each machine has initial knowledge of $O(|E(G)|\log n / k) = O(n^2 \log n /k)$ edges. 
Consider any machine $M_i$.
Since the random vertex partitioning and the sampling of the input graph are independent, there are at least
$2^{{n \choose 2}-O(n^2 \log (n) /k)}$ choices for the remaining edges, all of which are equally likely according to the random graph model, giving the claim.
\end{proof}
}

\onlyLong{
\begin{lemma}\label{lem:entropyUB}
Let $(\bp,r) \in \good$, where $\bp = (p_1,\dots,p_k)$. 
There exists a machine $M_i$ with output $\cA_i(\bp,r)$ such that, for every edge vector $z$ that has non-zero probability conditioned on $\cA_i(\bp,r)$, $p_i$, $r$,
$
\Prob{Z=z \mid \cA_i(\bp,r), p_i,r } \ge 1/2^{{n \choose 2} - O(n^2 \log(n)/k) - \Omega(({t}/{k})^{2/3})}$.
\end{lemma}

\begin{proof}
By assumption $(\bp,r) \in \good$, which means that the machines output all $t=\Theta({n \choose 3})$ triangles. Thus there is some machine $M_i$ that outputs at least $t/k$ triangles.
We will bound from below the number of edges known by machine $M_i$ conditioned on its output and its input assignment.

Initially, $M_i$ discovers $t_3=t_3(P_i)$ ``local'' triangles (for which it knows all $3$ edges) by inspecting its assigned portion of the input graph given by $P_i$.
Since we are restricting the inputs to be in $\good_i$, we know that the edges known to $M_i$ are bounded by $O(n^2 \log n /k)$ and hence the number of triangles formed using these edges is 
\[ 
  t_3 = O((n^2 \log n /k)^{3/2}) = O(n^3 \log^{3/2}(n) / k^{3/2}).
\]
We call a triangle $\lambda$ \emph{undetermined w.r.t.\ $M_i$}, if $M_i$ is unaware of at least one edge of $\lambda$ initially.
Formally, $\lambda$ is undetermined if there are two input graphs $G,G'$ where $\lambda$ exists in $G$ but not in $G'$ and both graphs are compatible with the input $p_i$ assigned to machine $M_i$.

By the above, we have at least $(t/k) - t_3$ undetermined triangles that are output by $M_i$.
From Equation~(10) in \cite{rivin2001counting}, we know that the number of distinct edges necessary for representing $\ell$ triangles is $\Omega(\ell^{2/3})$.
This means that at least $\left((t/k) - t_3\right)^{2/3}$ edges are required for representing the undetermined triangles of $M_i$.
We can divide the undetermined triangles into two sets, one set $T_1$ contains triangles that have vertex allocated to $M_i$, and the other set $T_2$ contains triangles that have no vertex allocated to $M_i$. Set $T_1$ contributes $|T_1|/(n \log n/k)$ unknown edges, since the number of vertices allocated to this machine is $O(n \log n/k)$, whereas $T_2$ contributes $1/3 \cdot (|T_2|)^{2/3}$ unknown edges. These two sets of unknown edges might overlap, hence we need to consider the maximum over them, which can be shown to be $\Omega(((t/k)-t_3)^{2/3})$. Hence it is possible to recover $\Omega(((t/k)-t_3)^{2/3})$ edges from $M_i$'s output that were unknown to $M_i$ initially. Let $\eta$ denote the number of unknown edges of $Z$ when $M_i$ outputs its solution.
Taking into account the initially known edges, we have
\begin{align}
  \label{eq:eta}
  \eta \le {n \choose 2} -  \Omega\left(\frac{t}{k} - t_3\right)^{2/3} - O\left(\frac{n^2\log n}{k}\right) = {n \choose 2} - O\left(\frac{n^2\log n}{k}\right) -  \Omega\left(\frac{t}{k}\right)^{2/3}
\end{align}
possible edges that are unknown to $M_i$, since $t_3 = o(t/k)$.
Since we have sampled the edges of the input graph following the $G_{n,1/2}$ random graph model, it follows that, for any $z$ that has nonzero probability given $M_i$'s output and initial assignment, 
$\Prob{Z=z\mid o_i, p_i,r} = 2^{-\eta}$.
The lemma follows by applying \eqref{eq:eta}.
\end{proof}
}

\onlyLong{
\paragraph{Proof of \Cref{thm:trianglelb}} We are now ready to instantiate \Cref{thm:expectedruntime} where $Z$ is the characteristic vector of edges as defined above.
Note that \Cref{lem:entropyLB1} and \Cref{lem:entropyUB} satisfy Premises~\eqref{eq:precond1} and \eqref{eq:precond2}. 
Note that $\Omega(\tfrac{t}{k})^{2/3} = \Omega(\tfrac{n^2}{k^{2/3}})$.
Setting $\infcost = \Theta\left(\tfrac{n^2}{k^{2/3}}\right)$
completes the proof of \Cref{thm:trianglelb}.
}

\paragraph{A tight lower bound in the congested clique}
Our analysis extends in a straightforward way to the congested clique model where, in a synchronous complete network of $n$ machines, every machine $u$ receives exactly one input vertex of the input graph and gets to know all its incident edges. 
Together with the deterministic upper bound of $O(n^{1/3})$ shown in~\cite{DolevLP12}, this implies the following:

\begin{corollary}\label{cor:congested}
The  round complexity of enumerating all triangles in the congested clique of $n$ nodes with high probability of success is $\Omega\left(\frac{n^{1/3}}{B}\right)$, assuming a link bandwidth of $B$ bits.
This bound is tight up to logarithmic factors.
\end{corollary}

\paragraph{Message lower bounds}
\onlyLong{We point out that it is possible to extend \Cref{thm:expectedruntime} to yield new message lower bounds for algorithms that attain an efficient time complexity. 
We outline the high-level argument for triangle enumeration.
Consider an algorithm matching the time bound of \Cref{thm:trianglelb}, i.e.,  $T=\tilde O(\tfrac{n^2}{k^{5/3}})$ assuming a bandwidth of $B=O(\log n)$ bits. 
In the $k$-machine, in $T$ rounds each machine can receive at most $\mu = \tilde O(\tfrac{n^2}{k^{2/3}})$ bits in total.
\Cref{lem:entropyLB1} tells us that every machine has very little initial knowledge about the $t$ triangles in the graph given its initial graph assignment, when considering inputs chosen from $\good$. 
On the other hand, inspecting the proof of \Cref{lem:entropyUB}, we can observe that
a machine $M_j$ who outputs $t_j$ triangles needs to receive $\tilde\Omega(t_j^{2/3})$ bits of information.
If we restrict the algorithm to terminate within $T$ rounds, this means that each machine can output at most $O(n^3 /k)$ triangles, as this requires $\mu = O((\tfrac{n^3}{k})^{2/3})$ bits of information. This implies that the output per machine must be roughly balanced and every machine needs to receive $\Omega(\mu)$ bits of information, yielding a message complexity of $\tilde\Omega(k\tfrac{n^2}{k^{2/3}}) = \tilde\Omega(n^2 k^{1/3})$. In particular, this rules out algorithms that aggregate all input information at a single machine (which would only require $O(m)$ messages in total). 
From the above, we have the following.}

\begin{corollary}\label{cor:msgtri}
 Let $\cA$ by any algorithm that enumerates all triangles with high probability and terminates in $\tilde O(\tfrac{n^2}{k^{5/3}})$ rounds. 
 Then, the total message complexity in the $k$-machine model of $\cA$ is $\tilde\Omega(n^2 k^{1/3})$. For $\tilde O(n^{1/3})$-rounds algorithms in the congested clique, the message complexity is $\tilde\Omega(n^{7/3})$.
\end{corollary}

%
%
%

%

%
%
%
%

%

%

\onlyShort{\vspace{-0.1in}}
\section{Upper Bounds}
\onlyShort{\vspace{-0.05in}}

\subsection{An Almost Optimal Algorithm for PageRank Approximation}\label{sec:pralgo}
\onlyShort{\vspace{-0.05in}}
In this section we present a simple distributed algorithm to approximate the $\pagerank$ vector of an input graph in the
$k$-machine model. This algorithm has a round complexity of $\tilde O(n/k^2)$, which significantly improves
over the previous $\tilde O(n/k)$-round solution~\cite{KlauckNPR15}.

We first recall the distributed random walk-based Monte-Carlo algorithm for computing $\pagerank$,
for a given reset probability $\epsilon$, as described in~\cite{pagerank}.
This algorithm is designed and analyzed in the standard {\cal CONGEST} model,
where each vertex of the graph executes the algorithm. The algorithm is as follows.
Initially, each vertex creates $c\log n$ random walk tokens, where $c=c(\epsilon)$ is a parameter defined in~\cite{pagerank}
($c(\epsilon)$ is inversely proportional to $\epsilon$), which are then forwarded according to the following process:
when a node $u$ receives some random walk token $\rho$, either, it terminates the token with probability $\epsilon$, or,
with probability $1-\epsilon$, forwards it to a neighbor of $u$ chosen uniformly at random.
Each machine keeps a variable $\psi_v$, for each of its nodes $v$, which counts the number of random walk tokens that
were addressed to $v$ (i.e., the total number  of all random walks that visit $v$).
Each node $v$ then estimates its $\pagerank$ by computing $\frac{\epsilon\psi_v}{cn\log n}$.
It can be shown that this estimate gives a $\delta$-approximation, for any constant $\delta > 0$,
to the $\pagerank$ value of each node $v$ with high probability, and that this algorithm terminates in $O(\log n/\epsilon)$
rounds with high probability~\cite{pagerank}. The key idea to obtain such a fast runtime is to send only the {\em  counts}
of the random walks, instead of keeping track of the random walks from different sources.
Clearly, only the number (i.e., count) of the random walks visiting a node at any step is required to estimate the $\pagerank$.
\onlyShort{In the full paper, we describe why a na\"ive implementation of the above approach only yields a running time of $\tilde O(n/k)$.}
\onlyLong{
Note that a straightforward implementation of the above random walk-based algorithm might yield a suboptimal running time in the $k$-machine model. (In fact, applying the Conversion Theorem of~\cite{KlauckNPR15} to implement the above algorithm
gives only $\tilde O(n/k)$ time.) The main issue is that some machine might receive too many random walks destined for the nodes in that machine. For example, during some step of the random
walk it might happen that $n$ different walks  are destined to different nodes in the same machine,
causing $\Omega(n)$ congestion at some machine leading to a $\Omega(n/k)$ bound. 
For example, in a star-like topology, the center vertex $c$ which resides at some machine $M_1$ might need to receive $n$ random walks from its neighbors, hence causing a round complexity of $\tilde\Omega(n/k)$. In the above example, since there is only one high degree vertex, we  can get around this problem  by sending only the counts.  However, the situation is less clear if $\Omega(n)$
tokens are destined for different nodes in the same machine.   
}

To avoid the above pitfalls, we describe an approach that directly exploits the $k$-machine model.
On the one hand, our goal is to reduce the total amount of communication while, on the other hand, we need to ensures that the incurred message complexity is balanced for the available machines. 
This motivates us to treat vertices differently depending on how many tokens they hold.
We say that a vertex $u$ is \emph{light in iteration $r$} if, conceptually, the machine that hosts $u$ considers less than $k$ tokens to be held at $u$.
Otherwise, we say that $u$ is \emph{heavy in iteration $r$}.
Note that, throughout the course of our algorithm, the value of $\tokens[v]$ depends on the topology of the input graph and hence a vertex can change its status w.r.t.\ being a heavy or a light vertex.

In our algorithm (Algorithm~\ref{alg:pagerank}), each machine $M$ stores an array $\tokens[u]$, which has an entry for each vertex $u$ hosted at $M$.
Initially, we generate $\Theta(\log n)$ tokens for each vertex which we use as the initialization value of $\tokens$.
Then, we mimic the (parallel) random walk steps of \cite{pagerank} by performing $\Theta(\log(n) / \epsilon)$ iterations where, in each iteration, each machine $M$ first considers the tokens stored for its light vertices. 
For each such token held at one of its vertices $u$, $M$ uniformly at random selects a neighboring vertex $v$ and keeps track of how many tokens have chosen $v$ in a separate array $\alpha[v]$.
In particular, $M$ also increments the same entry $\alpha[v]$ if $v$ is chosen as the destination for some token of a distinct low-load vertex $w \ne u$ at $M$.
Then, $M$ sends a message $\langle \alpha[v], \text{dest:} v \rangle$ for each $v$ where $\alpha[v]$ is nonzero, which is subsequently delivered to the destination machine using random routing (cf.\ Lemma~\ref{lem:direct-routing}). This ensures that all the messages are delivered in $\tilde O(n/k^2)$ rounds.

We now describe how high-load vertices are processed, each of which can hold up to $O(n\log n)$ tokens. 
To avoid potentially sending a large number of messages for a single high-load vertex $u$, machine $M$ considers the index set $I$ of machines that host at least one neighbor of $u$.
Then, for each token of $u$, machine $M$ samples an index from $I$ according to the degree distribution of $u$ (see Line~\ref{line:sampleDestMachine} in Algorithm~\ref{alg:pagerank}) and keeps track of these counts in an array $\beta$, which has an entry for each machine in $I$.
Finally, $M$ generates one message of type $\langle \beta[j], \text{src:} u \rangle$, for each entry $j$ where $\beta[j]>0$ and sends this count message directly to the respective destination machine. 
We show that these messages can be delivered in $\tilde O(n/k^2)$ rounds by proving that, with high probability, each machine holds $\tilde O(n/k^2)$ high-load vertices in any given iteration of the algorithm.

\begin{algorithm}[ht!]
\begin{algorithmic}[1]
  \State Let $V_i$ denote the vertices hosted by machine $M_i$
  \State Initalize array $\tokens[u] \gets \lceil c\log n \rceil$, for $u \in V_i$, where $c>0$ is a suitable constant \Comment{$\tokens[u]$ represents the current number of tokens at vertex $u$}
  \For{$\Theta(\log(n) / \epsilon )$ iterations}
  \For{$u \in V_i$}
    \State sample $t$ from distribution $Binomial(\tokens[u],\epsilon)$
    \State $\tokens[u] \gets \tokens[u] - t$ \label{line:terminate} \Comment{Terminate each token with probability $\epsilon$}
  \EndFor
  \State
  \State Initialize array $\alpha[v] \gets 0$, for each $v \in V$ 
  \Comment{Process the light vertices} \label{line:lowload1}
  \For{each vertex $u \in V_i$ where $\tokens[u] < k$}
      \State let $N_u \subseteq V$ be the set of neighbors of vertex $u$
      \While{$\tokens[u]>0$}
        \State sample $v$ uniformly at random from $N_u$
        \State $\alpha[v]  \gets \alpha[v] + 1$
        \State $\tokens[u] \gets \tokens[u] - 1$
      \EndWhile
  \EndFor 
  \For{each $v \in V_i$ where $\alpha[v] > 0$}
    \State \textbf{send} message $\langle \alpha[v], \text{dest: }v \rangle$ \textbf{to} the machine hosting vertex $v$ using random routing \label{line:lowLoadMsg}
  \EndFor   \label{line:lowload2}
  \State
  \For{each vertex $u \in V_i$ where $\tokens[u] \ge k$} \label{line:highload1}
  \Comment{Process the heavy vertices} 
      \State let $I \subseteq [k]$ be the index set of the machines that host a neighbor of $u$
      \State initialize array $\beta[j] \gets 0$, for each $j \in I$
      \While{$\tokens[u]>0$}
        \State let $n_{j,u}$ be number of neighbors of $u$ hosted at machine $M_j$ and let $d_u$ be $u$'s degree
        \State sample index $j$ from distribution $\left(\frac{n_{1,u}}{d_u},\dots,\frac{n_{k,u}}{d_u}\right)$ \Comment{Note $\sum_{j=1}^k n_{j,u} = d_u$} \label{line:sampleDestMachine}
        \State $\beta[j] \gets \beta[j] + 1$
        \State $\tokens[u] \gets \tokens[u] - 1$
      \EndWhile
  \For{each $j \in I$ where $\beta[j]>0$}
    \State \textbf{send} message $\langle \beta[j], \text{src: }u \rangle$ \textbf{to} machine $M_j$ \label{line:highLoadMsg}
  \EndFor 
  \EndFor \label{line:highload2}
  \State
  \For{each received message of type $\langle c_w, \text{dest: }w \rangle$}
    \State $\tokens[w] \gets \tokens[w] + c_w$
  \EndFor
  \For{each received message of type $\langle c_v, \text{src: }v \rangle$} \label{line:highrcv1}
    \While{$c_v>0$} 
      \State let $N_v \subseteq V$ be the set of neighbors of $v$ hosted at $M_i$
      \State sample $w$ uniformly at random from $N_v$
      \State $\tokens[w] \gets \tokens[w] + 1$
      \State $c_v \gets c_v - 1$
    \EndWhile \label{line:highrcv2}
  \EndFor 
  \EndFor
\end{algorithmic}
\caption{Approximating the PageRank with reset probability $\epsilon>0$. Code for machine $M_i$.}
\label{alg:pagerank}
\end{algorithm}

\begin{proposition} 
  \label{lem:correctness}
  Algorithm~\ref{alg:pagerank} correctly computes the \pagerank\ with high probability. 
\end{proposition}

\begin{proof}
In \cite{pagerank} it is shown that the random walk process, where each token is either terminated with probability $\epsilon$ or forwarded with probability $1 - \epsilon$ to a neighbor chosen uniformly at random, approximates the \pagerank\ of the graph.
Thus it is sufficient to show that Algorithm~\ref{alg:pagerank} adheres to this random walk process.

Consider a node $u$ and suppose that $u$ holds $\ell$ tokens. 
If $\ell < k$, then according to Lines~\ref{line:lowload1}-\ref{line:lowload2}, we increment the corresponding entry of array $\alpha[v]$, for some uniformly at random chosen neighbor $v$ of $u$ and send a message $\langle c_v, dest: v \rangle$ to the machine $M'$ hosting $v$.
Upon receiving the message, $M'$ increases its token count of $v$, as required.

Now, suppose that $\ell \ge k$ and consider an arbitrary neighbor $v$ of $u$, hosted on machine $M'$ and assume that $M'$ hosts $n_u \ge 1$ neighbors of $u$ in total. 
For any token of $u$, it follows from Line~\ref{line:sampleDestMachine} that we choose machine $M'$ with probability $\frac{n_u}{d_u}$, where $d_u$ is the degree of $u$ in the graph. 

The algorithm then sends a message of type $\langle c_u, src:u \rangle$ to machine $M'$ where $c_u$ is the number of tokens of $u$ for which $M'$ was sampled as the destination machine.
Upon processing this message in Lines~\ref{line:highrcv1}-\ref{line:highrcv2}, $M'$ delivers each token to its locally hosted neighbors of $u$ uniformly at random, and hence a specific neighbor $v$ receives a token with with probability $\frac{1}{n_u}$.

Combining these observations, we conclude that $v$ receives a token with probability $\frac{n_u}{d_u} \frac{1}{n_u} = \frac{1}{d_u}$, conditioned on the token not having been terminated in Line~\ref{line:terminate} with probability $\epsilon$, which corresponds to the random walk process of \cite{pagerank}.
\end{proof} 

\begin{lemma} \label{lem:sendOut}
  Every machine $M_i$ sends at most $O(n\log(n)/k)$ messages in any 
  iteration $r$ with high probability.
\end{lemma}

\begin{proof}
First, we consider messages that $M_i$ needs to send on behalf of its hosted light vertices.
We classify the light vertices into \emph{send bins} $S_0,S_1,\dots,S_{\lceil\log k\rceil -1}$, according to the number of distinct messages that they require to be sent and, for each $j$, $0 \le j \le \lceil\log_2 k\rceil - 1$, we define the bin
\begin{equation}
  \label{eq:binBnd}
  S_j = \left\{ v \in V(G)\ \middle|\ \frac{k}{2^{j+1}} \le \tokens[v] < \frac{k}{2^j} \right\}.
\end{equation}
By definition, the total number of messages generated for any light vertex in iteration $r$ is at most $k-1$, and hence every light $v$ is in some bin $S_j$.

Since $\Theta(\log n)$ tokens are generated initially for each vertex, we have $\Theta(n\log n)$ tokens in total, which implies that $|S_j| \le \frac{2^{j+1}n\log n}{k},$ for all $j$.
By the random vertex partitioning, we know that a machine $M_i$ receives at most $ O(|S_j|\log(n)/k)$ vertices from $S_j$ with probability $\ge 1 - n^{-4}$; we denote this vertex set by $S_{i,j}$.
Taking a union bound over the iterations of the algorithms (assuming a constant reset probability $\epsilon$), the $O(\log_2 k)$ distinct bins, and over the $k$ machines, it follows that 
\begin{equation} \label{eq:bins} 
  \forall M_i\ \forall j \in \{0,\dots,\lceil\log_2 k\rceil-1\} \colon |S_{i, j }| = O\left(\frac{2^{j+1}n\log n}{k^2}\right),
\end{equation}
with probability $\ge 1 - n^{-2}$.
According to \eqref{eq:binBnd}, each vertex in bin $S_j$ holds less than $k/2^j$ tokens, and thus by \eqref{eq:bins} the total number of messages produced by vertices in $S_j$ that are located on machine $M_i$ is 
\[
  O\left(|S_{i,j}|\cdot\frac{k}{2^j}\right) 
  =  O\left(\frac{2^{j+1}n\log n}{k^2}\frac{k}{2^j}\right) 
=  O\left(\frac{n\log n}{k}\right).
\]
Since we have $\Theta(\log k)$ bins, the total number of messages generated by machine $M_i$ for its light vertices is $ O(n\log(n)/k) \cdot \Theta(\log k) = \tilde O(n/k)$ with high probability. 

Now, consider the heavy vertices at $M_i$.
By definition, each heavy vertex has at least $k$ tokens and hence there are at most $O(n\log(n)/k)$ heavy vertices at any point of the algorithm.
Therefore, the random vertex partitioning implies that each machine will hold most $O(n\log(n)/k^2)$ many heavy vertices whp.
For processing the tokens of a heavy vertex $u$, we recall from Algorithm~\ref{alg:pagerank} that we need to send at most $1$ message to each machine that holds a neighbor of $u$. 
This means that all messages generated for $u$ can be sent and delivered in $1$ round and hence by taking a union bound over all the machines, it follows that each machine can send all tokens for its heavy vertices in $O(n\log(n)/k^2)$ rounds.

Finally, the lemma follows by taking a union bound over the $O(\log(n) / \epsilon)$ iterations of the algorithm.
\end{proof}

A key ingredient in the analysis of the algorithm is the following simple lemma, which quantifies
how fast some specific routing can be done in the $k$-machine model.

\begin{lemma}\label{lem:direct-routing}
Consider a complete network of $k$ machines, where each link can carry one message of $O(\polylog n)$
bits at each round. If each machine is source of $O(x)$ messages whose destinations are distributed
independently and uniformly at random, or each machine is destination of $O(x)$ messages whose
sources are distributed independently and uniformly at random, then all the messages can be routed
in $O((x \log x)/k)$ rounds w.h.p.
\end{lemma}

\onlyLong{
\begin{proof}
We shall prove the statement for the case in which each machine is the source of $O(x)$ messages.
The other case and its analysis are symmetric.

Since destinations of messages are chosen randomly, we choose to route each message to its
(random) destination machine through the link that directly connects the source to the destination
machine (which always exists because the network is complete). By a classic balls-into-bins result,
each of the $k-1$ links of each machine is responsible for carrying $O((x \log x)/k)$ messages w.h.p.,
and the result follows.
\end{proof}
}

\begin{lemma} 
  \label{lem:sendTime}
  Consider any iteration $r$ of Algorithm~\ref{alg:pagerank}. Then, with high probability, all messages generated at iteration $r$
  can be delivered in $\tilde O(n/k^2)$ rounds.
\end{lemma}

\begin{proof}
We first consider the messages generated due to a heavy vertex $u$. 
Recall from Algorithm~\ref{alg:pagerank} that each machine directly sends the messages that it generated for $u$ to the destination machines, which requires just $1$ round.  
As we have argued in Lemma~\ref{lem:sendOut}, there are at most $O(n\log(n)/k^2)$ many heavy vertices per machine whp and hence all of their messages can be delivered within $O(n\log(n)/k^2)$ rounds.

In the remainder of the proof, we focus on messages generated while processing light vertices.
To this end, we argue that each machine needs to receive at most $\tilde O(n/k)$ messages that were generated due to light vertices in Line~\ref{line:lowLoadMsg}, which according to the random routing result, can be delivered in $\tilde O(n/k^2)$ rounds. 
We proceed similarly to the analysis in Lemma~\ref{lem:sendOut}. 
That is, we define \emph{receive bins} $R_0,R_1,\dots,R_{\lceil\log k\rceil-1}$, where 
\[
R_j = \left\{ v \in V(G) \mid \frac{k}{2^{j+1}} \le \lambda_v \le \frac{k}{2^j} \right\}\]
 and $\lambda_v$ is
the random variable that counts the number of tokens generated for light vertices that are received by $v$ in iteration $r$.
Consider any $v \in V(G)$ located at some machine $M$.
The crucial point is that each $v$ must be in exactly one of these bins, since Line~\ref{line:lowLoadMsg} ensures that machine $M$ receives at most $1$ message of type $\langle \alpha[v], \text{dest:} v\rangle$  that is addressed to $v$ from each distinct machine $M'$.

Similarly as in Lemma~\ref{lem:sendOut}, it follows by the properties of the random vertex partitioning that each machine holds $\tilde O(|R_j|/k)$ vertices from $R_j$ with high probability,
and hence the total number of messages that each machine needs to receive (over all  receive bins) is $\tilde O(n/k)$. 
Thus, by Lemma~\ref{lem:direct-routing}, all of these messages can be delivered in $\tilde O(n/k^2)$ rounds.
Finally, it is shown in \cite{pagerank} that all tokens are terminated in $O(\log (n)/\epsilon)$ steps and thus, assuming that $\epsilon>0$ is a small constant,
the claim follows by a union bound over the iterations of the algorithm.

\end{proof}

From Lemma~\ref{lem:sendTime} we conclude that all messages generated in a single iteration of Algorithm~\ref{alg:pagerank} can be delivered in $\tilde O(n/k^2)$ rounds with high probability. A union bound implies the following result.

\begin{theorem}
Algorithm \ref{alg:pagerank} computes a $\delta$-approximation of the PageRank vector of
an $n$-node graph in the $k$-machine model with high probability in $\tilde O(n/k^2)$ rounds,
for any constant $\delta > 0$.
\end{theorem}

\onlyShort{\vspace{-0.1in}}
\subsection{An Almost Optimal Algorithm for Triangle Enumeration}\label{sec:trialgo}
\onlyShort{\vspace{-0.05in}}

In this section we present a randomized algorithm that enumerates all the triangles
of an input graph $G = (V,E)$, and that terminates in $\tilde O(m/k^{5/3} + n/k^{4/3})$ rounds w.h.p.
This bound does not match the (existential) $\tilde \Omega(m/k^{5/3})$ lower bound
provided in Section~\ref{sec:trianglelb} only for very sparse graphs.

Our algorithm is a generalization of the algorithm \textsf{TriPartition} of Dolev et al.\ for the
congested clique model~\cite{DolevLP12}, with some crucial differences explained next.
The key idea, which in its generality can be traced back to~\cite{AfratiU11}, is to partition
the set $V$ of nodes of $G$ in $k^{1/3}$ subsets of $n/k^{1/3}$ nodes each,
and to have each of the $k$ machines to examine the edges between pairs of subsets
in one of the $(k^{1/3})^3 = k$ possible triplets of subsets (repetitions are allowed).

The algorithm is as follows. Each node picks independently and uniformly at random one color from a set
$C$ of $k^{1/3}$ distinct colors through a hash function $h:V \rightarrow C$ initially known by all the machines.
This gives rise to a color-based partition of the vertex set $V$ into $k^{1/3}$ subsets of $\tilde O(n/k^{1/3})$ nodes each, w.h.p.
A deterministic assignment of triplets of colors, hard-coded into the algorithm, logically assigns each of the $k$
possible triplets of such subsets to one distinct machine.
Each machine then collects all the edges between pairs of subsets in its triplet.
This is accomplished in two steps: (1) For each of the edges it holds,
each machine designates one random machine (among the $k$ machines) as the \emph{edge proxy} for that edge, and sends all its
edges to the respective edge proxies. The designation of an edge itself is done 
by the following {\em proxy assignment rule}
(this is necessary to avoid congestion at any one machine): A machine that has a node $v$
whose degree is at least $2k \log n$ requests all other machines to designate the
respective edge proxies for each of the incident edges of node $v$. If two machines
request each other to designate the same edge (since their endpoints are hosted
by the respective machines), then such a tie is broken randomly.
(2) In the second step, all the machines collect their required edges from the respective proxies: since
each edge proxy machine knows the hash function $h$ as well as the deterministic assignment of triplets,
it can send each edge to the machines where it is needed.
Then, each machine simply enumerates all the triangles in its local subgraph.

Our algorithm differs from the one in~\cite{DolevLP12} in the way the $k^{1/3}$
subsets of vertices are constructed, in the use of proxy computation and in the routing of
messages, which in our algorithm is randomized and hence requires a more involved analysis,
allowing for a better time complexity for graphs where the number of edges $m$ is $o(n^2)$.

We now argue that the above algorithm correctly enumerates all the triangles of a graph $G$,
and analyze its round complexity. A key step in the analysis of the complexity is to bound
from above the number of edges assigned to each machine. Observe that the number of edges
between pairs of subsets of one triplet is no larger than the number of edges in the subgraph
of $G$ induced by the nodes of one triplet; in turn, because of the random color-based partition
of the vertices made by the algorithm, the latter quantity is asymptotically equivalent to the
number of edges in the subgraph of $G$ induced by a set of (in this case, $\tilde O(n/k^{1/3})$)
randomly-chosen nodes of a graph. Thus, we shall concentrate on the latter quantity (which is of interest
in its own right). To this end, we will use the following concentration result due to R{\"{o}}dl 
and Ruci\'{n}ski~\cite{RodlR94}.\footnote{Observe
that one cannot simply apply a Chernoff bound, since edges are not chosen independently; also, mimicking the
argument for the proof of Lemma~4.1 in~\cite{KlauckNPR15} would give a bound of the form $\tilde O(m/k^{1/3})$,
which is weaker since we would be overcounting edges (as we would be counting also those edges with just one
endpoint in the given machine).}

\begin{proposition}[{\cite[Proposition~1]{RodlR94}}]\label{pro:machine_edges}
Let, for a graph $G = (V,E)$, $m < \eta  n^2$, and let $R$ be a random subset of $V$ of size $|R| = t$ such that $t \geq 1/3\eta$.
Let $e(G[R])$ denote the number of edges in the subgraph induced by $R$. Then,
\[
\emph{Pr}\left[e(G[R]) > 3 \eta t^2\right] < t \cdot e^{-ct}
\]
for some $c > 0$.\onlyLong{\footnote{A careful
inspection of the argument used by R{\"{o}}dl and Ruci\'{n}ski to establish this result reveals that the additional condition
$t \geq 1/3\eta$, missing from their statement, is necessary for the result to hold. In fact, as stated, their result
is implicitly assuming that both $n$ and $t$ grow to infinity~\cite{personal}.}}
\end{proposition}

We are now ready to analyze the algorithm.
\onlyShort{Lack of space forces us to defer the analysis to the full version of the paper.}

\begin{theorem}
There is a distributed algorithm for the $k$-machine model that enumerates all the triangles of an $n$-node,
$m$-edge graph in $\tilde O(m/k^{5/3} + n/k^{4/3})$ rounds with high probability.
\end{theorem}

\onlyLong{
\begin{proof}
Since there are $(k^{1/3})^3 = k$ possible triplets of non-intersecting subsets of $n/k^{1/3}$ nodes,
all possible triangles are examined by the algorithm, and this proves its correctness.

We now argue that the algorithm terminates in $\tilde O(m/k^{5/3} + n/k^{4/3})$ rounds w.h.p.
As part of the argument used to prove Lemma~4.1 of~\cite{KlauckNPR15} it is shown that every
machine initially stores $\tilde O(m/k + \Delta)$ edges, where $\Delta$ is the maximum degree of the graph. 
If we apply Lemma~\ref{lem:direct-routing} directly, the communication phase that assigns the edges to their
random proxies takes $\tilde O(m/k^2 + \Delta/k)$ rounds w.h.p. We now show that the proxy assignment
rule allows us show an $\tilde O(m/k^{5/3})$ bound for this phase for every non-sparse graph.

Clearly, by the random proxy assignment, each machine receives only $\tilde{O}(m/k)$ messages.
We next argue that each machine is responsible for designating only $\tilde{O}(m/k)$ edges w.h.p.
Then, by Lemma~\ref{lem:direct-routing}, the time to send all the designation messages
is $\tilde{O}(m/k^2)$ rounds.

For the sake of the analysis, we partition the non-isolated nodes of the input graph into $\log n$ sets, based on their degree:
the $i$-th set contains all the nodes whose degree is in $[\Delta/2^i, \Delta/2^{i+1})$, $0 \leq i \leq \log n - 1$.
We now focus on the number of messages sent by some machine $M$.  By a standard Chernoff bound, a node $v_i$
with degree $d_i$ in the $i$-th set has $\tilde{O}(d_i/k)$ neighbors in $M$ w.h.p. If $n_i$ is number of nodes
in the $i$-th set, then the total number of neighbors (and hence messages) that $M$ will send with respect
to nodes in this set is $\tilde{O}(n_id_i/k)$ w.h.p. Summing over all the $\log n$ sets we have that
the total number of messages sent by $M$ is $\sum_{i=0}^{\log n - 1} \tilde{O}(n_id_i/k) = \tilde{O}(m/k)$ w.h.p.\
(via the union bound). Applying the union bound over all the machines, we have that the same bound holds for every machine. 

The above argument does not take into account the messages sent by a machine  initially to request
designation of an edge. A machine needs one round (to broadcast to all the other machines) to request such a designation.
If some machine $M$ sends $f \geq k \polylog n$ requests, then 
$M$ must have $f$ nodes with degree at least $2k \log n$.  By the RVP, this implies that with high probability  the total number of
nodes with degree at least $2k \log n$ is at least $\Omega(fk)$. Hence the number of edges in the graph is $m = \tilde{\Omega}(fk^2)$.
Therefore the number of rounds needed for broadcast, $\tilde{O}(f)$, is subsumed
by $\tilde O(m/k^{5/3})$.

Next we analyze the re-routing of each edge $e$ from its edge proxy to all the machines that are assigned a copy of both of
the endpoints of $e$. Observe that any two nodes, and therefore any edge, can be held by at most $k^{1/3}$
different machines: consider an edge $(a,b)$, and pick one machine $M$ that has to receive it because among its
three subsets of nodes, one (call it $A$) contains $a$ and one (call it $B$) contains $b$. Edge $(a,b)$ can be
assigned only to those machines which contain \emph{both} subsets $A$ and $B$, and there are only $k^{1/3}-1$
such machines in addition to $M$. Hence, re-routing the edges entails $m k^{1/3}$ messages to be traveling across
the network.\footnote{Notice that each node is replicated $k^{2/3}$ times in the system, and therefore each edge
is replicated $k^{4/3}$ times; however, we only need to re-route copies of edges that are \emph{internal} to the triplets,
and therefore copies of edges that have one endpoint in one triplet and the other endpoint in a different triplet need
not be communicated. Hence, the total number of edges to be communicated is $m k^{1/3}$ and not $m k^{2/3}$.}
We first bound the number of edges received by each machine. Fix one machine $M$.
\onlyShort{ Let, for a graph $G = (V,E)$, $m < \eta  n^2$, and let $R$ be a random subset of $V$ of size $|R| = t$ such that $t \geq 1/3\eta$, and suppose that
$e(G[R])$ denotes the number of edges in the subgraph induced by $R$. Then, it was shown in~\cite{RodlR94} that
$
\emph{Pr}\left[e(G[R]) > 3 \eta t^2\right] < t \cdot e^{-ct}
$
for some $c > 0$.
}
\onlyLong{We shall apply Proposition~\ref{pro:machine_edges}}\onlyShort{ We apply this result} with $t = d n \log n/k^{1/3}$ for some positive
constant $d$. We have two cases. If $m \geq  n k^{1/3}/6d\log n$ then $m \geq n^2/6t$, which in turn implies $2m/n^2 \geq 1/3t$,
and thus we can apply Proposition~\ref{pro:machine_edges} with $\eta = 2m/n^2$ obtaining, for machine $M$,
\onlyLong{\[}
  \onlyShort{$}
\Prob{e(G[R]) > 3 \frac{2m}{n^2} \mleft(\frac{dn \log n}{k^{1/3}}\mright)^2} < t e^{-c d n \log n /k^{1/3}},
\onlyShort{$}
\onlyLong{\]}
that is, since $k \leq n$,
\onlyLong{\[}
\onlyShort{$}
\Prob{e(G[R]) \leq \frac{6 d^2 m\log^2 n}{k^{2/3}}} > 1 - e^{-\Omega\mleft(\log n\mright)}.
\onlyShort{$}
\onlyLong{\]}
Hence we can apply Lemma~\ref{lem:direct-routing} with $x = \tilde O(m/k^{2/3})$, which yields a round
complexity of $\tilde O(m/k^{5/3})$ w.h.p.
Now observe that each proxy has to send $\tilde O(m/k^{2/3})$ edges. We can apply Lemma~\ref{lem:direct-routing}
with $x = \tilde O(m/k^{2/3})$, which implies that the number of rounds needed for the proxies to send their edges is
$\tilde O(m/k^{5/3})$ w.h.p., completing the analysis for the case $m \geq  n k^{1/3}/6d\log n$.

On the other hand, if $m <  n k^{1/3}/6d\log n$ we shall apply Proposition~\ref{pro:machine_edges} with
$\eta = 1/3t = k^{1/3}/3dn\log n$, obtaining
\[
\Prob{e(G[R]) > 3 \frac{k^{1/3}}{3dn\log n} \mleft(\frac{dn \log n}{k^{1/3}}\mright)^2} < t e^{-c d n \log n /k^{1/3}},
\]
that is, since $k \leq n$,
\[
\Prob{e(G[R]) \leq \frac{dn\log n}{k^{1/3}}} > 1 - e^{-\Omega\mleft(\log n\mright)}.
\]
As in the previous case, we apply Lemma~\ref{lem:direct-routing}, now with $x = \tilde O(n/k^{1/3})$. The theorem follows.
\end{proof}
}

\section{Conclusions}

We presented a general technique for proving lower bounds on the round complexity of distributed computations in a general message-passing model for large-scale computation, and showed its application for two prominent graph problems, \pagerank and triangle enumeration. We also presented near-optimal algorithms
for these problems, which can be efficiently implemented in practice.

Our lower bound technique works by relating the size of the output to the number of communication rounds needed, and could be useful in showing lower bounds for other problems where the output size is large (significantly more than the number of machines), such as sorting, matrix multiplication, shortest paths, matching, clustering, and densest subgraph.

\bibliographystyle{abbrv}
\bibliography{biblio}

\end{document}